\documentclass[graybox]{svmult}

\usepackage{type1cm}        
\usepackage{makeidx}         
\usepackage{graphicx}        
\usepackage{multicol}        
\usepackage[bottom]{footmisc}

\usepackage{epsf,amscd,amsmath,amssymb,amsfonts,mathtools,url}

\makeindex             

\newcommand{\setdiff}{\! \setminus \!}

\theoremstyle{plain}
\newtheorem{thm}{Theorem}[section]
\newtheorem{cor}[thm]{Corollary}

\theoremstyle{definition}

\theoremstyle{remark}
\newtheorem{rem}[thm]{Remark}

\numberwithin{equation}{section}

\newcommand{\be}{\begin{equation}}
\newcommand{\ee}{\end{equation}}
\newcommand{\bea}{\begin{eqnarray}}
\newcommand{\eea}{\end{eqnarray}}
\newcommand{\beas}{\begin{eqnarray*}}
\newcommand{\eeas}{\end{eqnarray*}}

\theoremstyle{plain}
\theoremstyle{definition}

\newtheorem{prop}[thm]{Proposition}
\newtheorem{lem}[thm]{Lemma}

\newtheorem{ex}[thm]{Example}

\numberwithin{thm}{section}
\numberwithin{equation}{section}

\def\One{\mathbb{I}}

\def\trianglerg{\;\raisebox{-4mm}{\includegraphics[width=9mm]{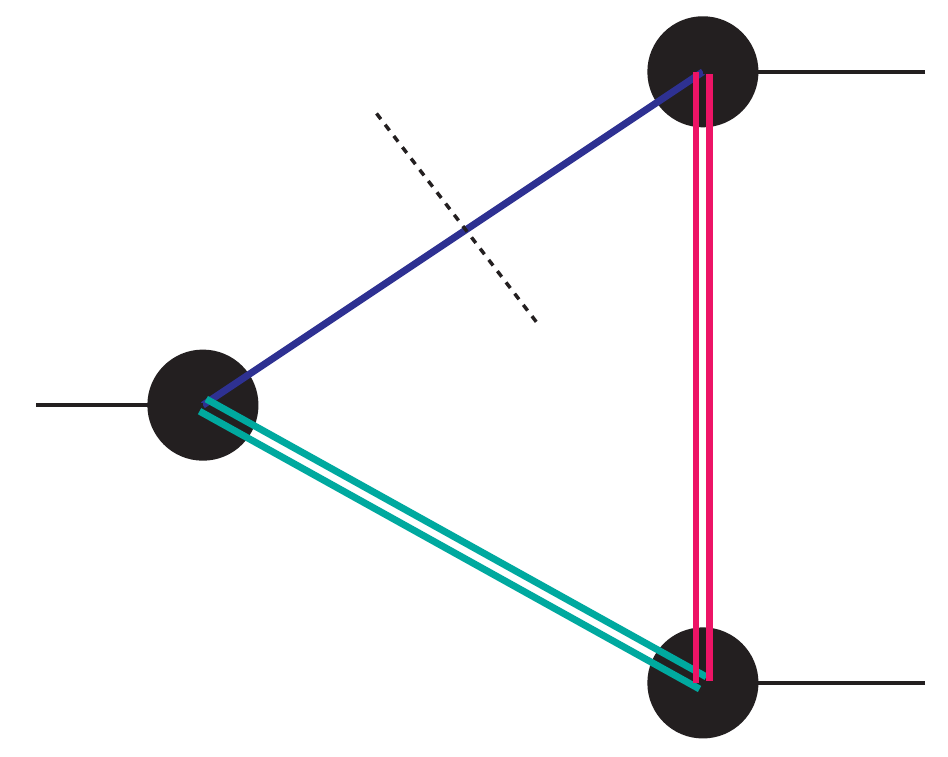}}\;}
\def\triangleg{\;\raisebox{-4mm}{\includegraphics[width=9mm]{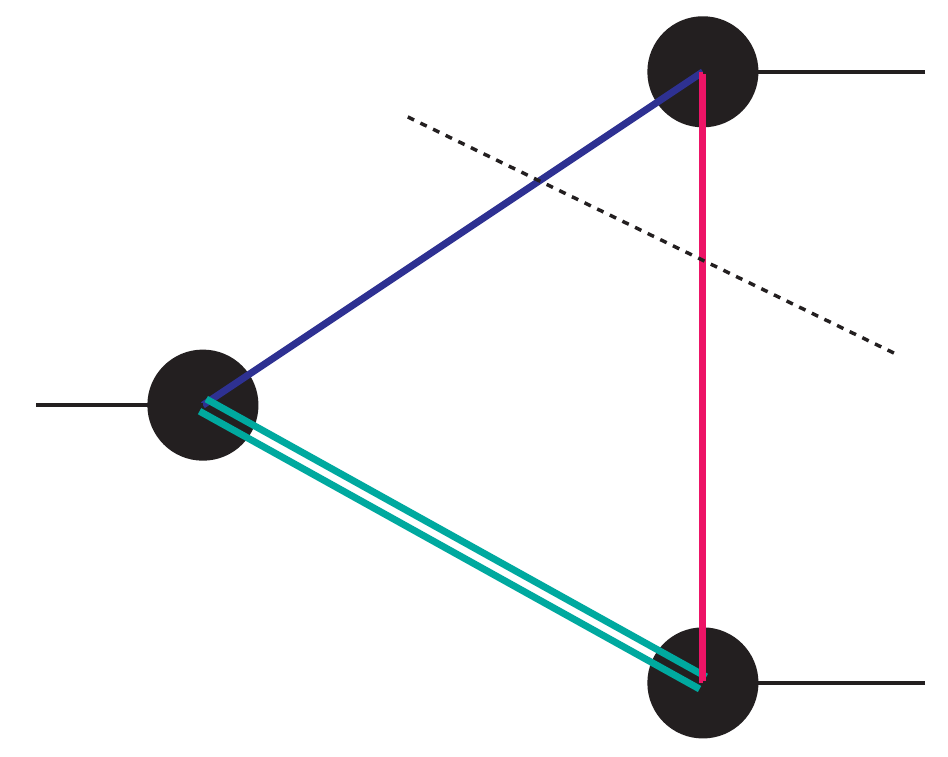}}\;}
\def\triangler{\;\raisebox{-4mm}{\includegraphics[width=9mm]{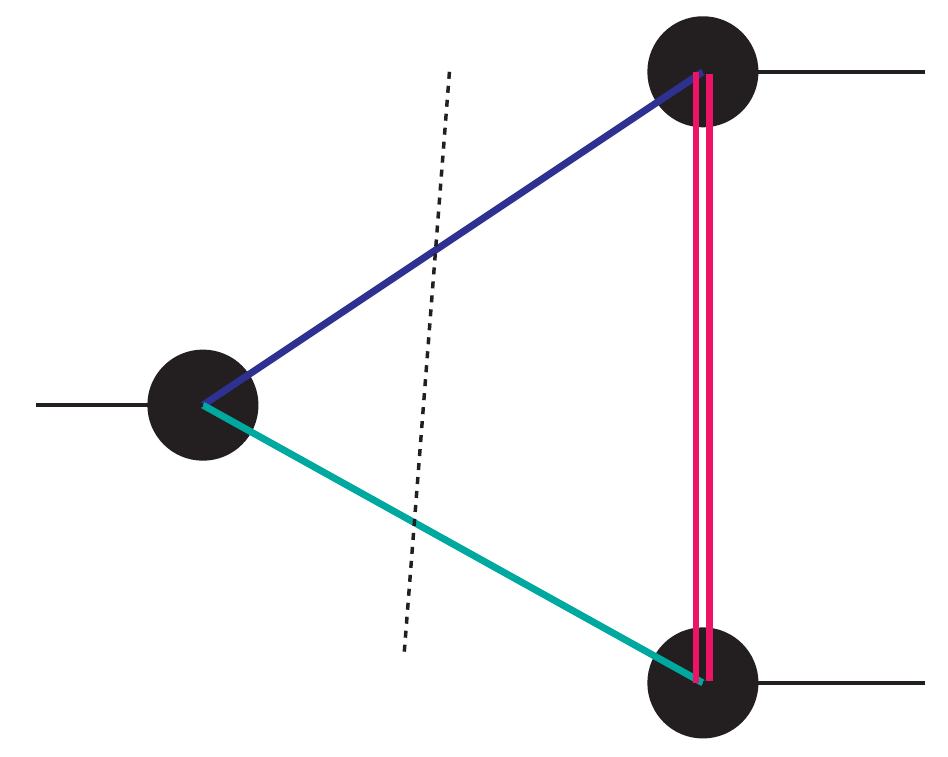}}\;}
\def\triangleb{\;\raisebox{-4mm}{\includegraphics[width=9mm]{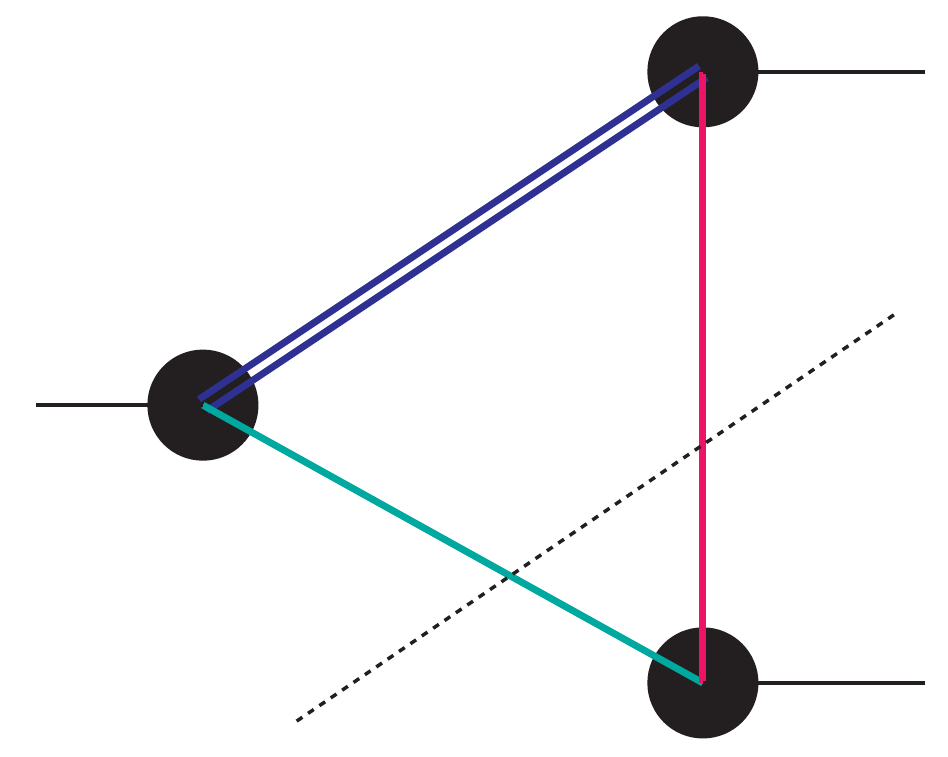}}\;}
\def\trianglebr{\;\raisebox{-4mm}{\includegraphics[width=9mm]{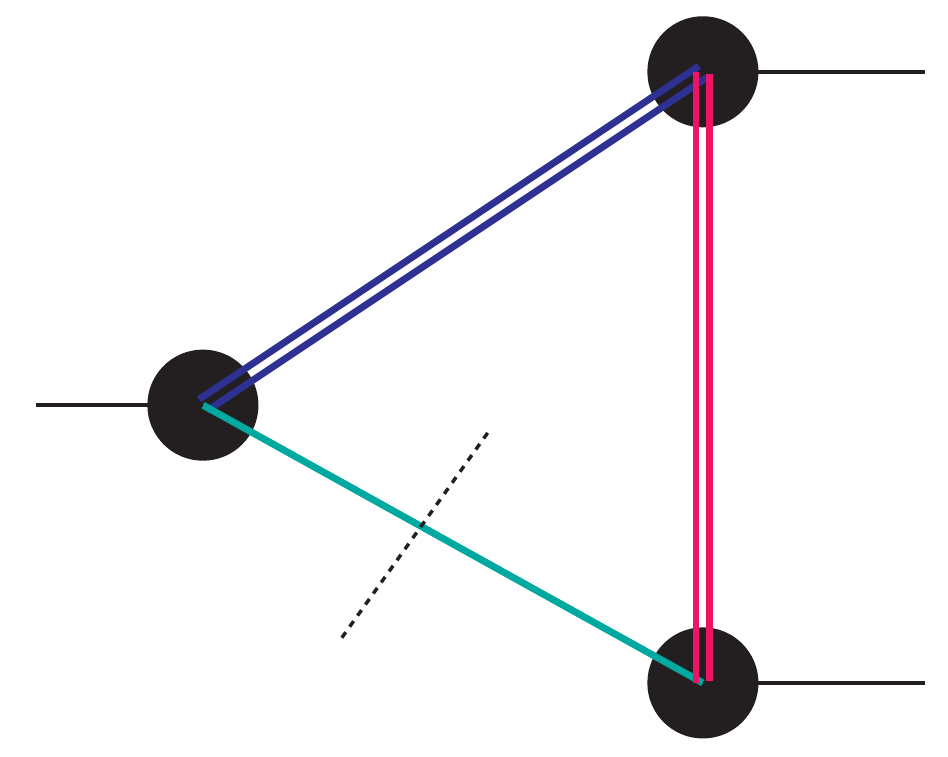}}\;}
\def\trianglegb{\;\raisebox{-4mm}{\includegraphics[width=9mm]{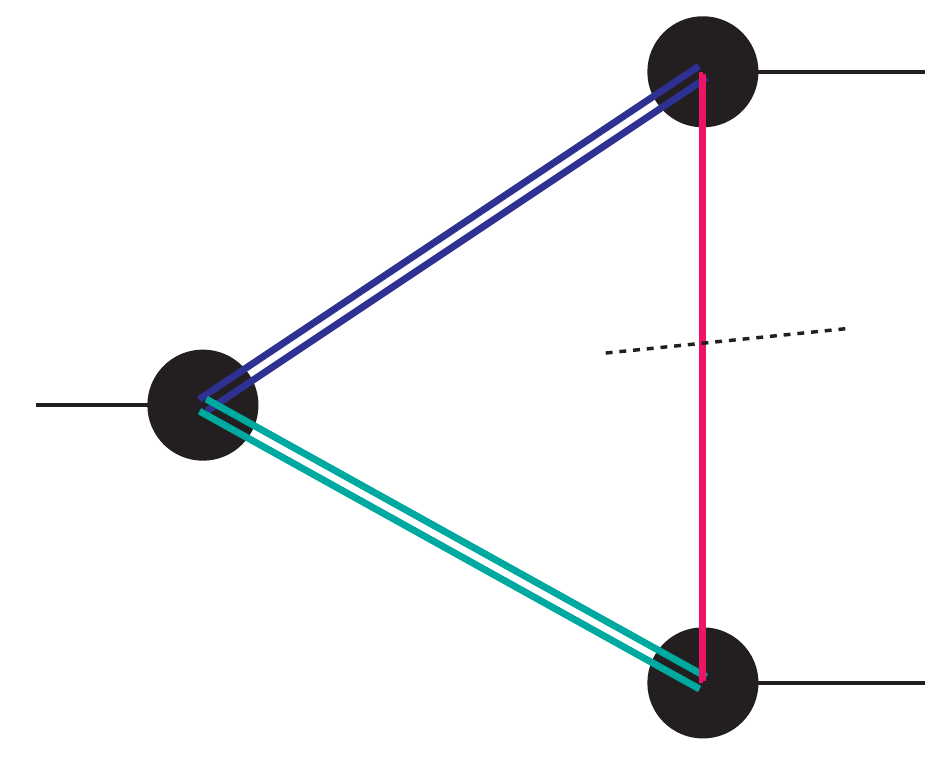}}\;}
\def\triangle{\;\raisebox{-4mm}{\includegraphics[width=9mm]{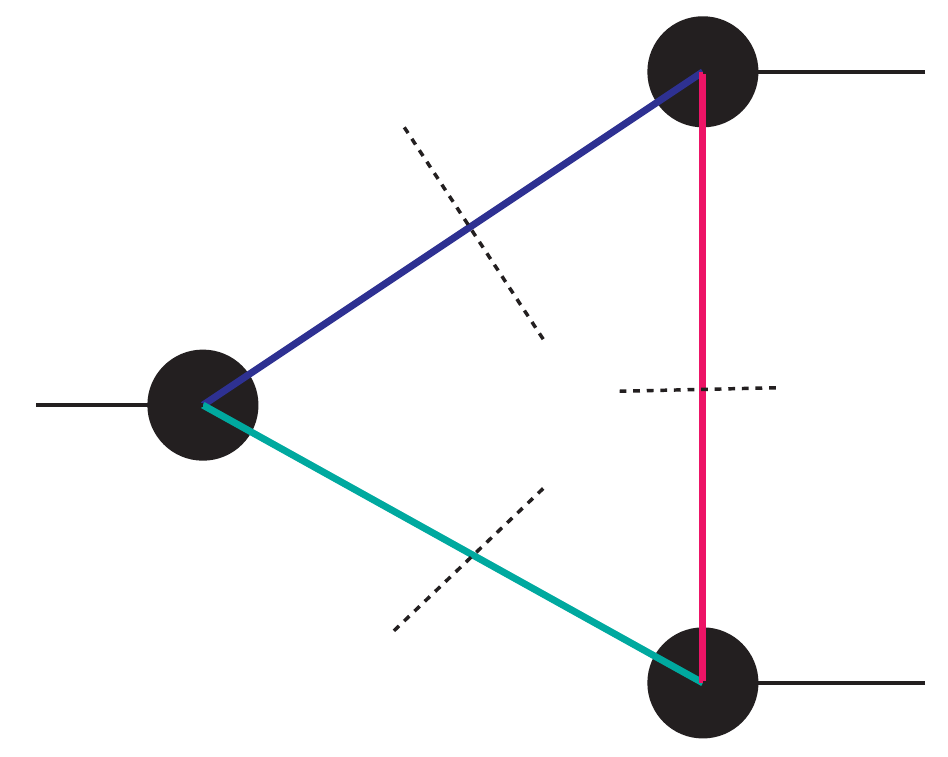}}\;}
\def\trianglebrg{\;\raisebox{-4mm}{\includegraphics[width=9mm]{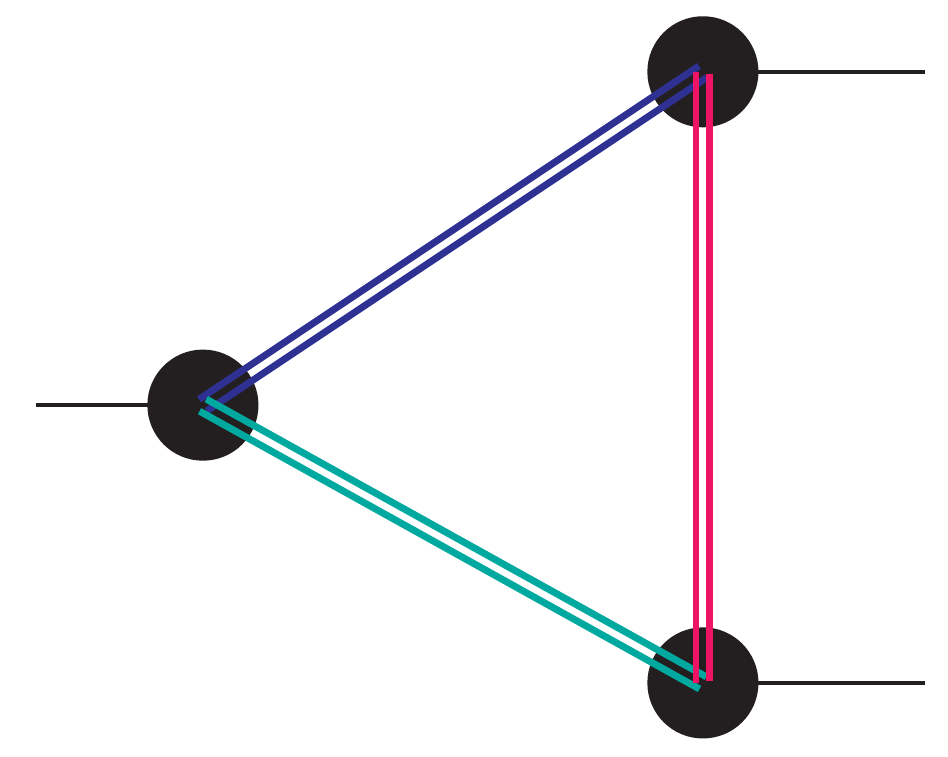}}\;}
\def\trianglebrr{\;\raisebox{-4mm}{\includegraphics[width=9mm]{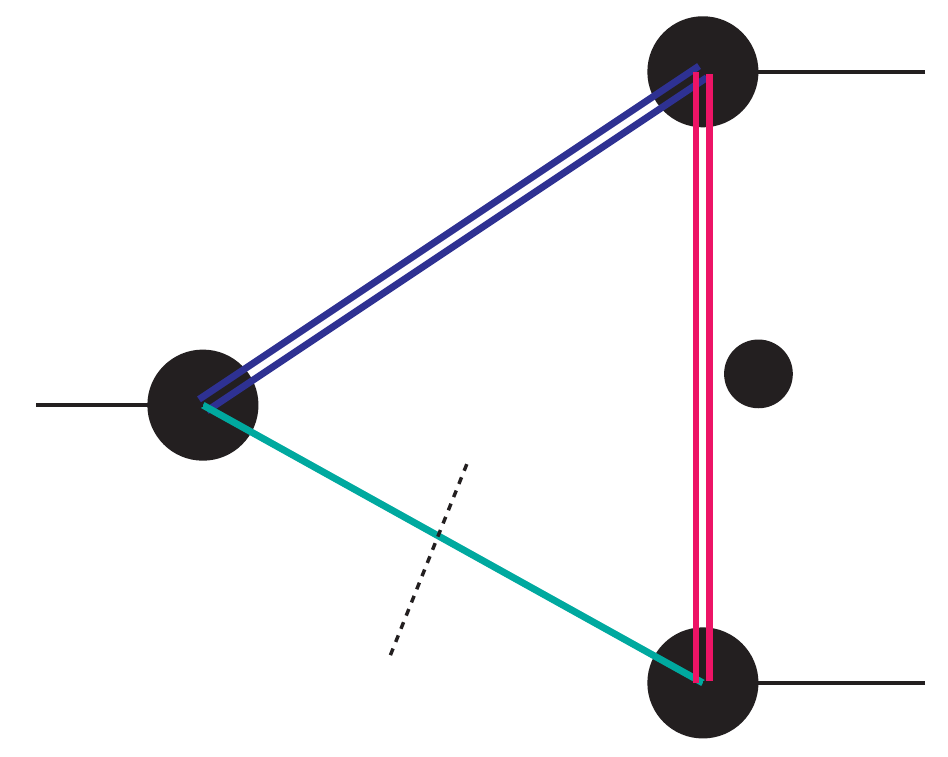}}\;}
\def\trianglebrrg{\;\raisebox{-4mm}{\includegraphics[width=9mm]{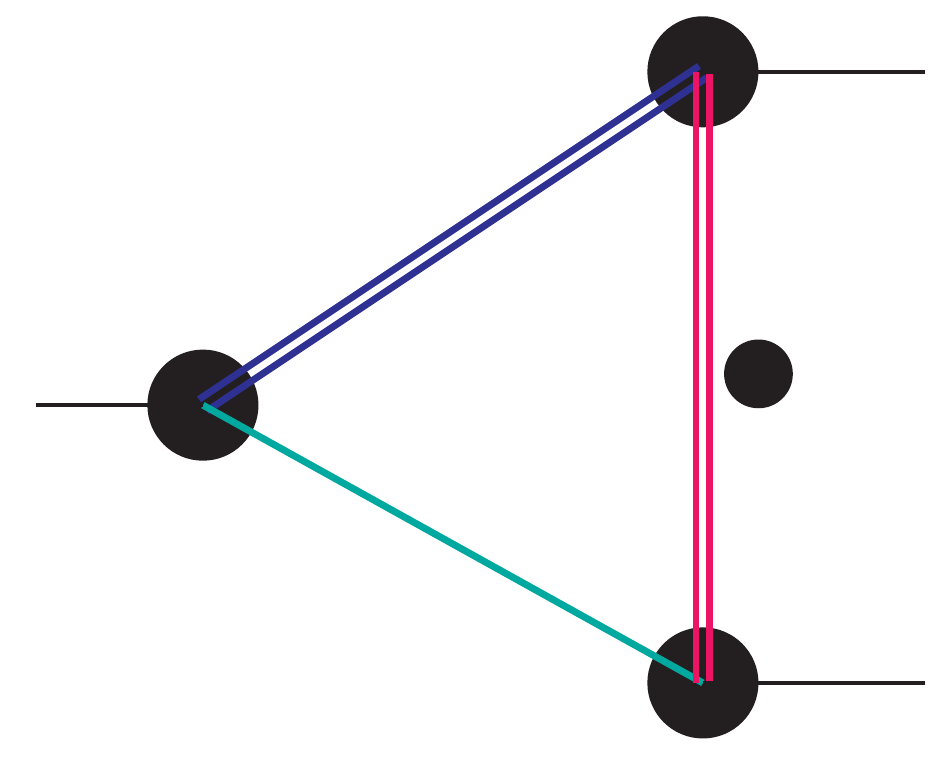}}\;}
\def\trianglebrb{\;\raisebox{-4mm}{\includegraphics[width=9mm]{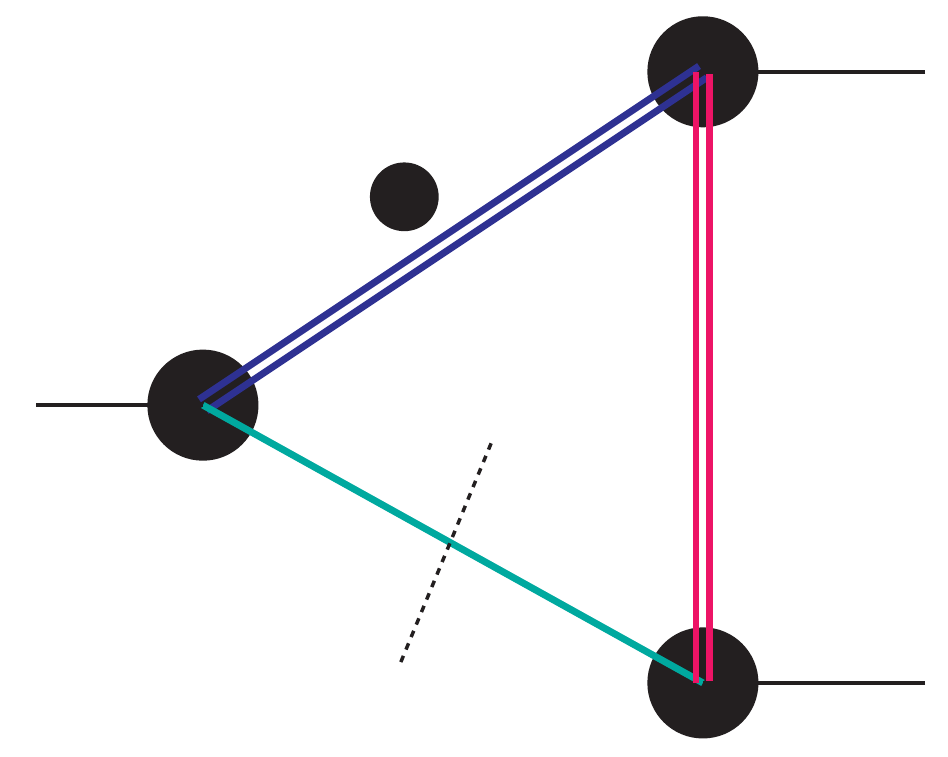}}\;}
\def\trianglergg{\;\raisebox{-4mm}{\includegraphics[width=9mm]{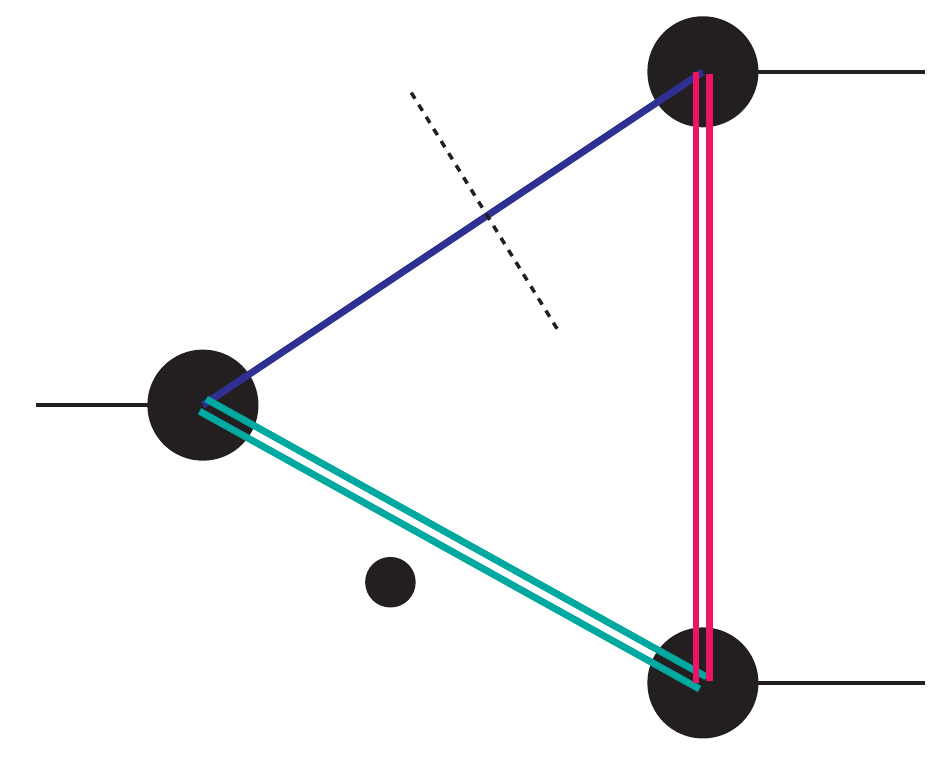}}\;}
\def\trianglerggb{\;\raisebox{-4mm}{\includegraphics[width=9mm]{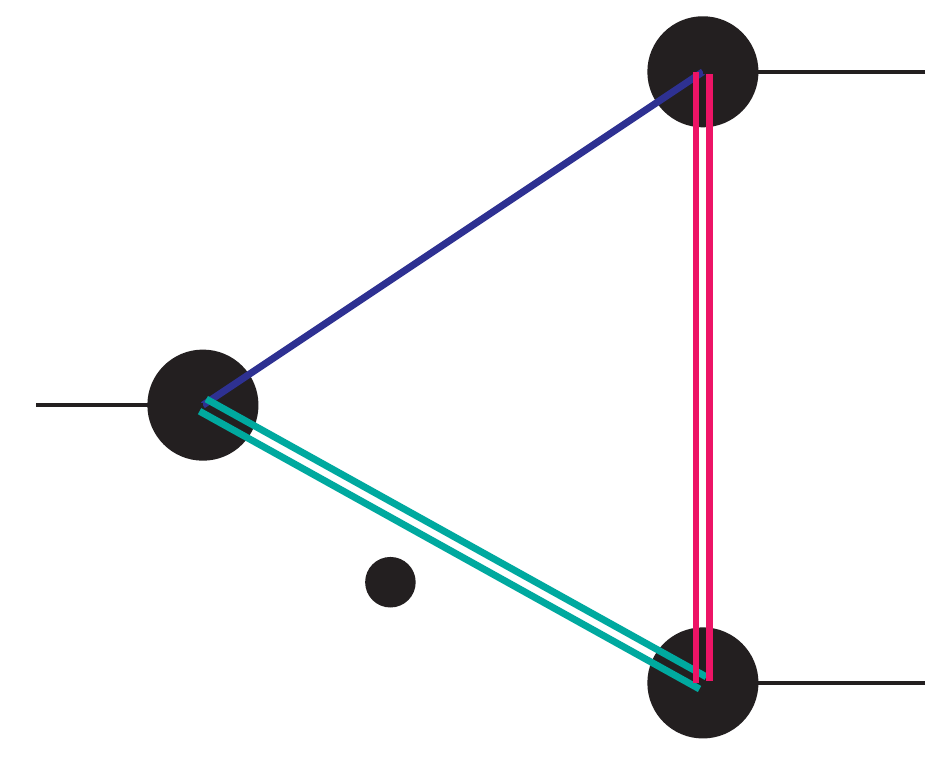}}\;}
\def\trianglergrb{\;\raisebox{-4mm}{\includegraphics[width=9mm]{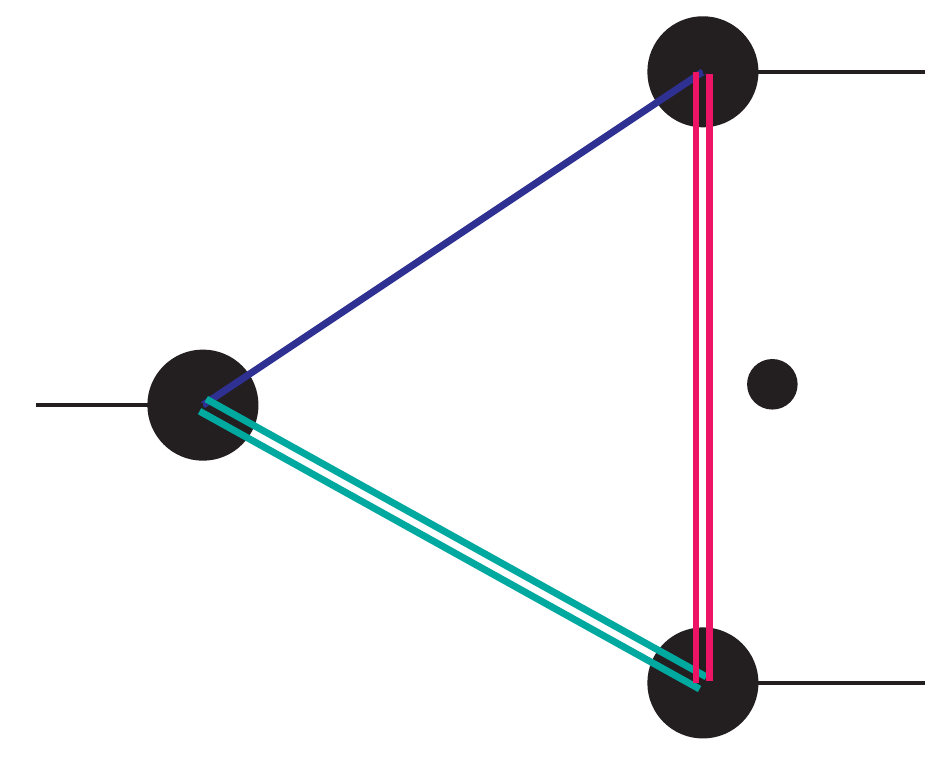}}\;}
\def\trianglergr{\;\raisebox{-4mm}{\includegraphics[width=9mm]{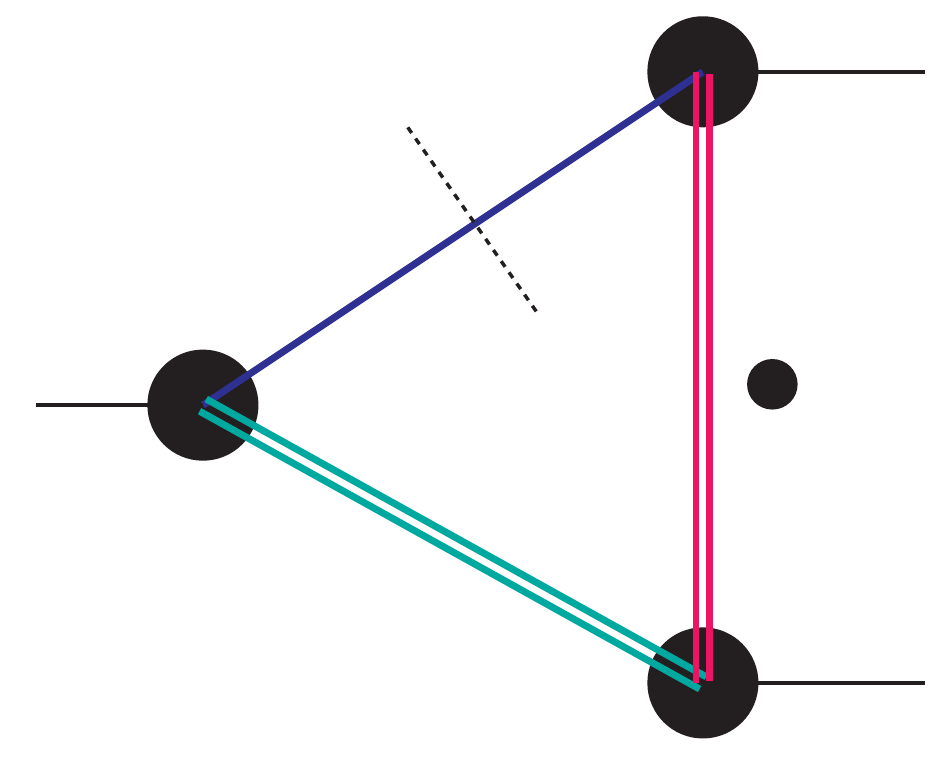}}\;}
\def\trianglegbg{\;\raisebox{-4mm}{\includegraphics[width=9mm]{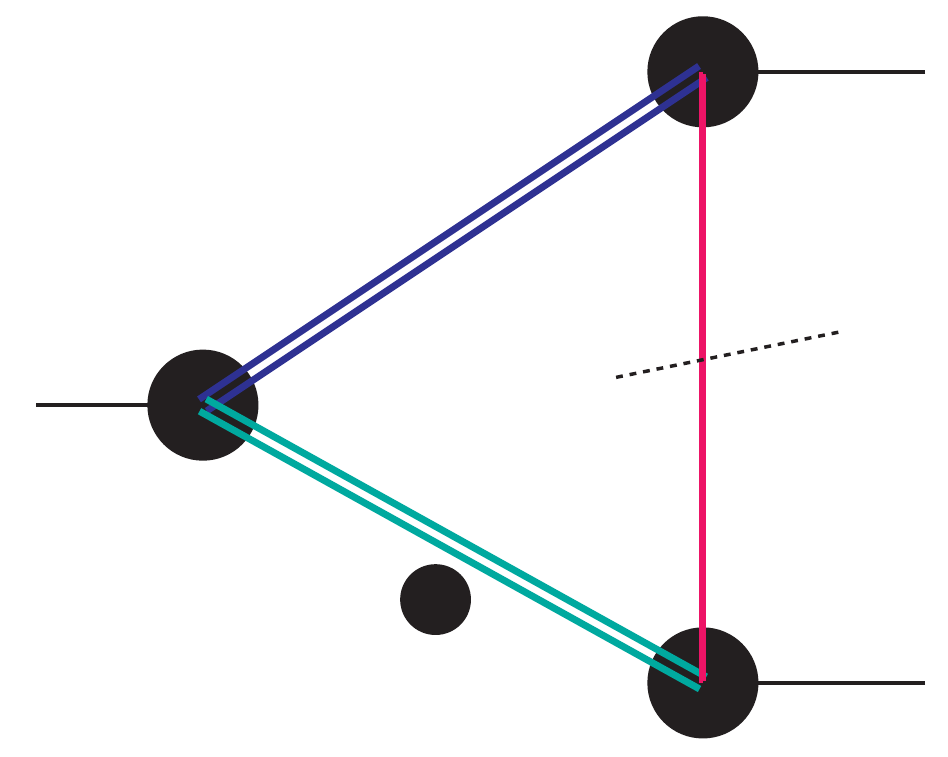}}\;}
\def\trianglegbgr{\;\raisebox{-4mm}{\includegraphics[width=9mm]{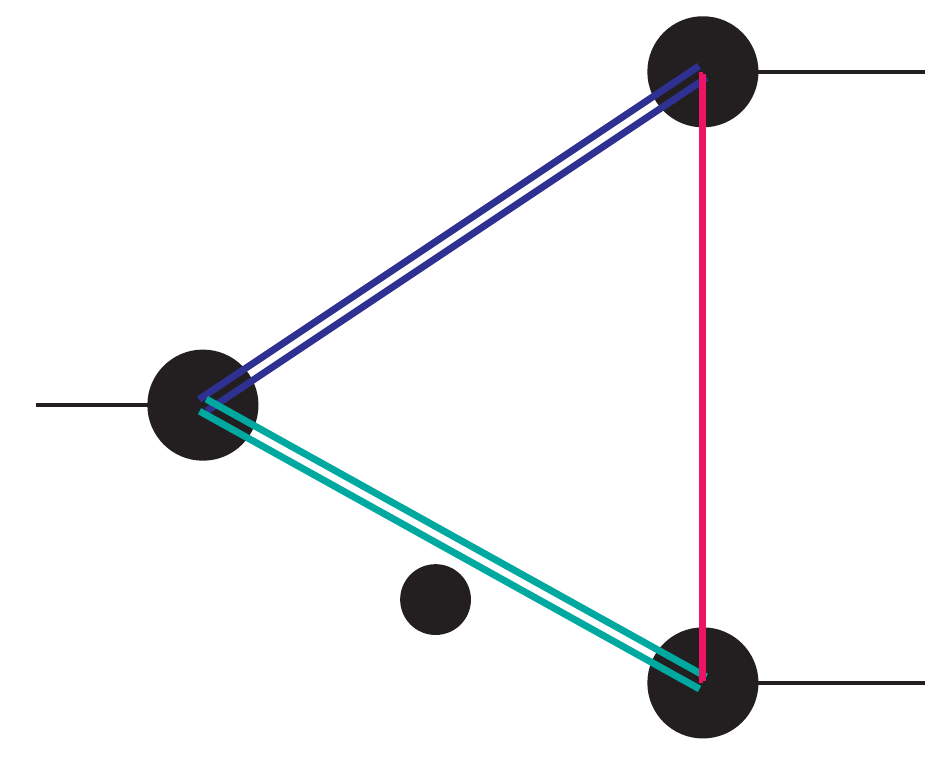}}\;}
\def\trianglegbbr{\;\raisebox{-4mm}{\includegraphics[width=9mm]{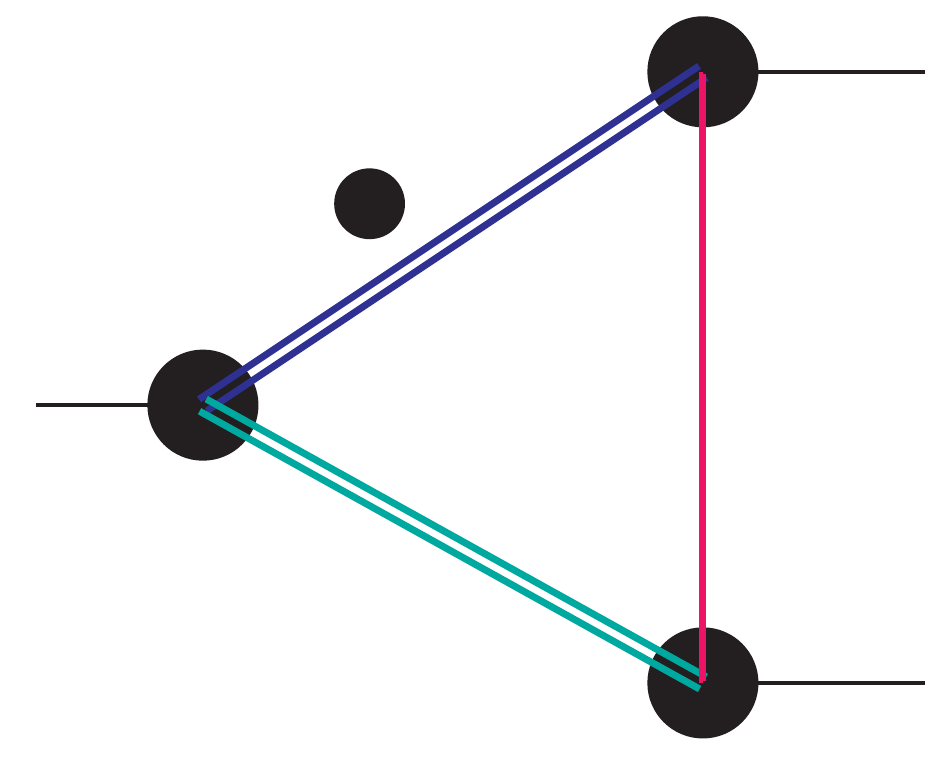}}\;}
\def\trianglegbb{\;\raisebox{-4mm}{\includegraphics[width=9mm]{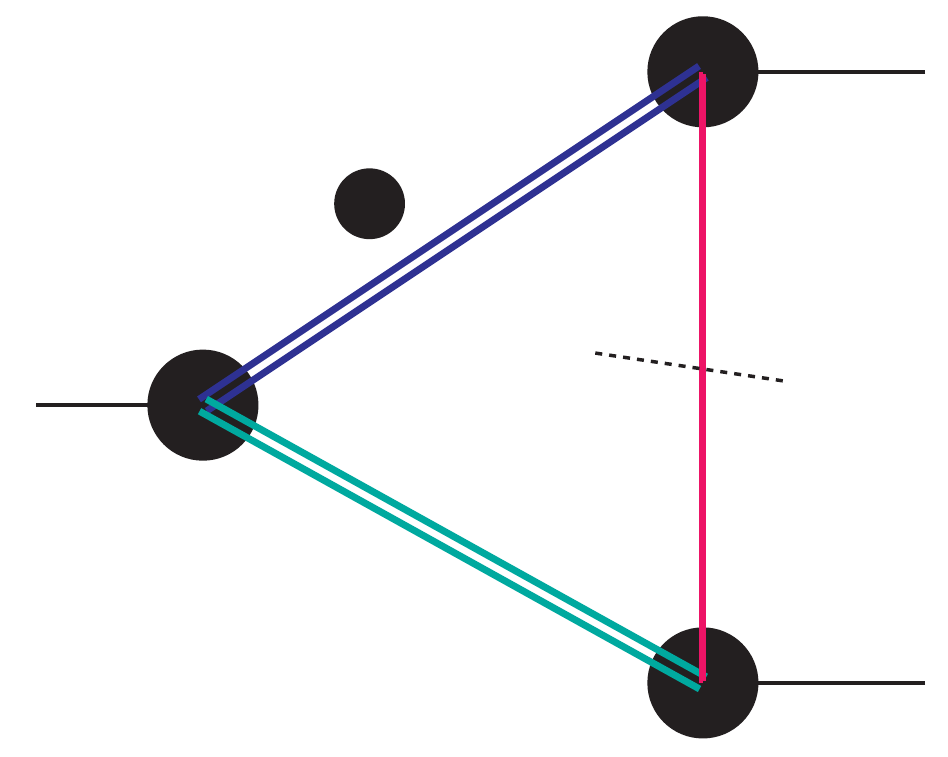}}\;}
\def\trianglebrbg{\;\raisebox{-4mm}{\includegraphics[width=9mm]{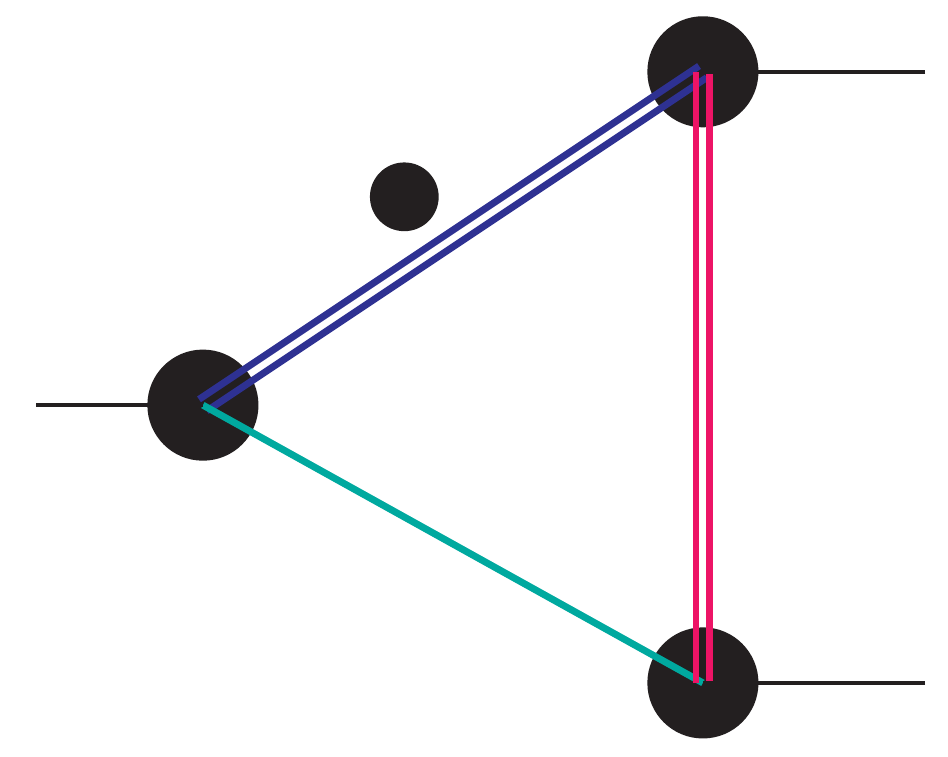}}\;}
\def\tadpoleb{\;\raisebox{-4mm}{\includegraphics[width=4mm]{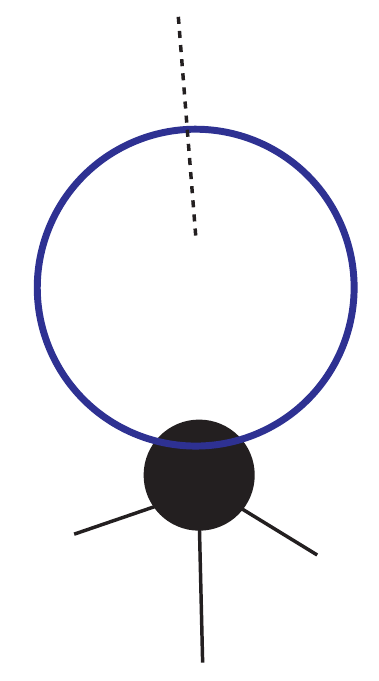}}\;}
\def\tadpoler{\;\raisebox{-4mm}{\includegraphics[width=4mm]{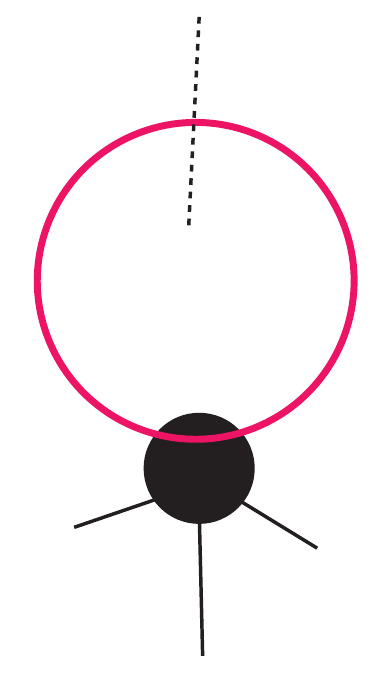}}\;}
\def\tadpoleg{\;\raisebox{-4mm}{\includegraphics[width=4mm]{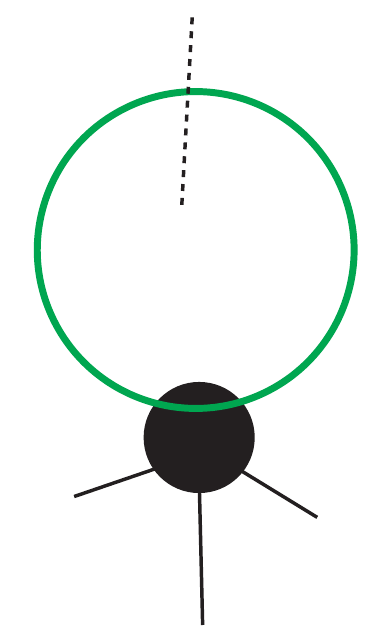}}\;}
\def\tadpolebb{\;\raisebox{-4mm}{\includegraphics[width=4mm]{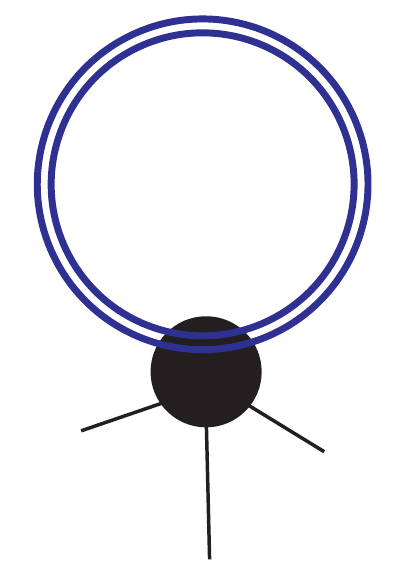}}\;}
\def\tadpolebbb{\;\raisebox{-4mm}{\includegraphics[width=4mm]{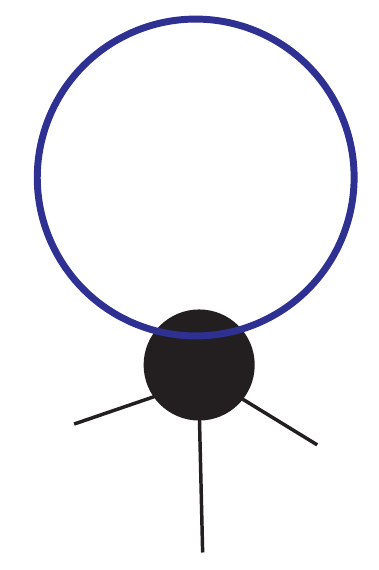}}\;}
\def\tadpolerr{\;\raisebox{-4mm}{\includegraphics[width=4mm]{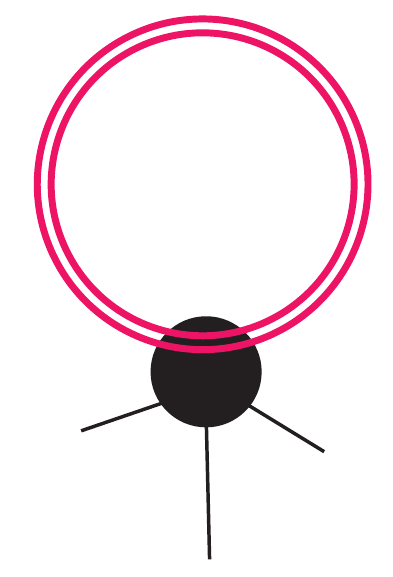}}\;}
\def\tadpolerrr{\;\raisebox{-4mm}{\includegraphics[width=4mm]{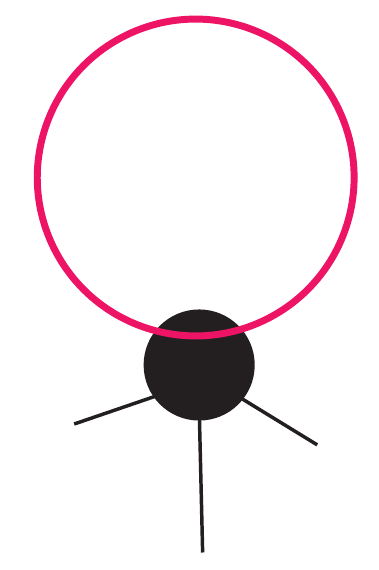}}\;}
\def\tadpolegg{\;\raisebox{-4mm}{\includegraphics[width=4mm]{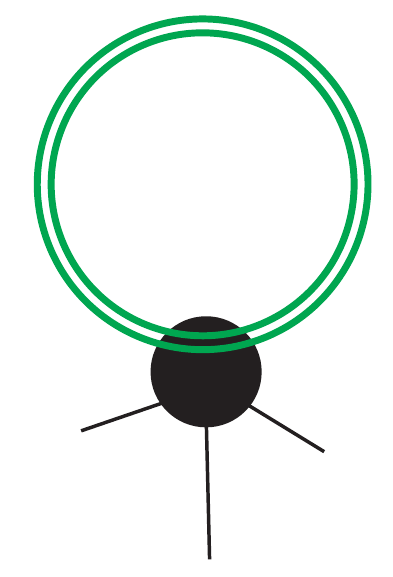}}\;}
\def\tadpoleggg{\;\raisebox{-4mm}{\includegraphics[width=4mm]{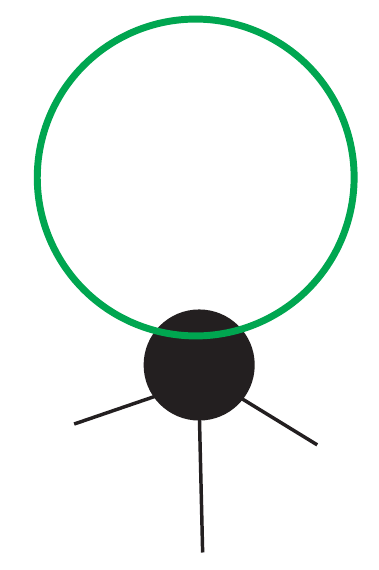}}\;}
\def\bubblegb{\;\raisebox{-5mm}{\includegraphics[width=7mm]{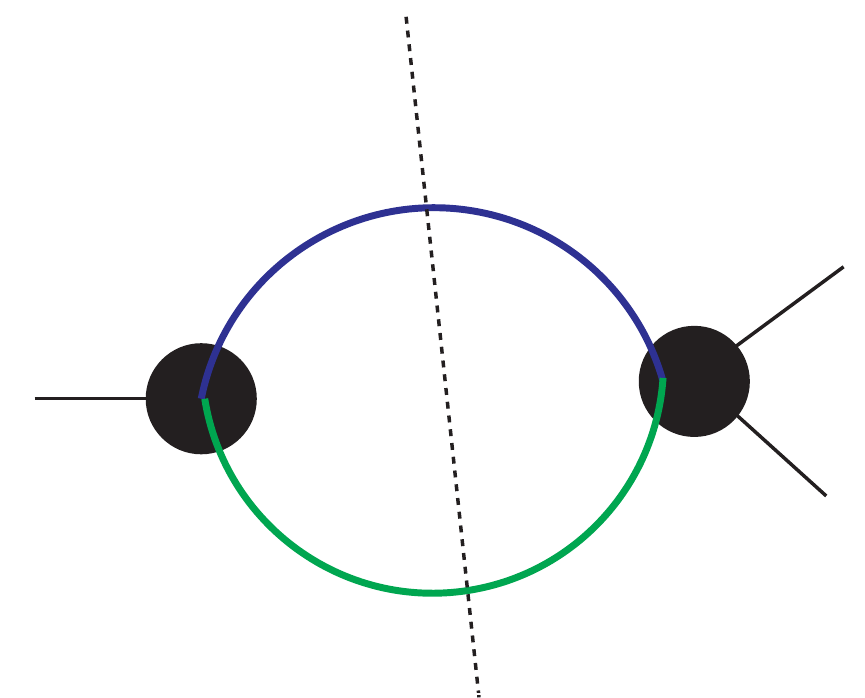}}\;}
\def\bubblegbgb{\;\raisebox{-5mm}{\includegraphics[width=7mm]{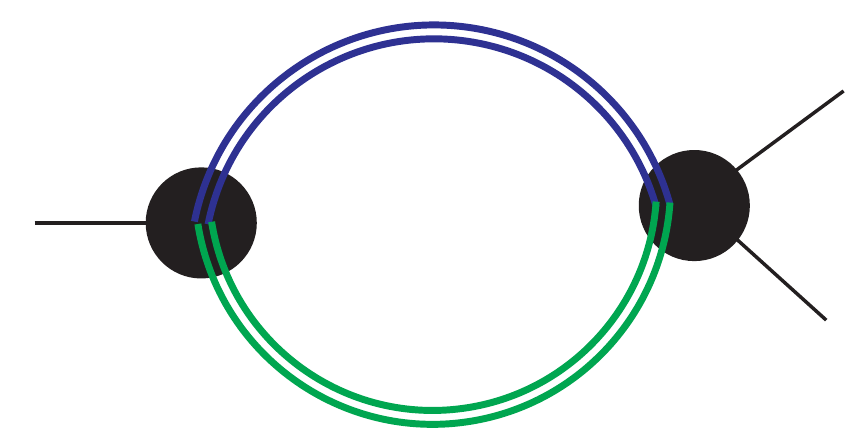}}\;}
\def\bubblegbgg{\;\raisebox{-5mm}{\includegraphics[width=7mm]{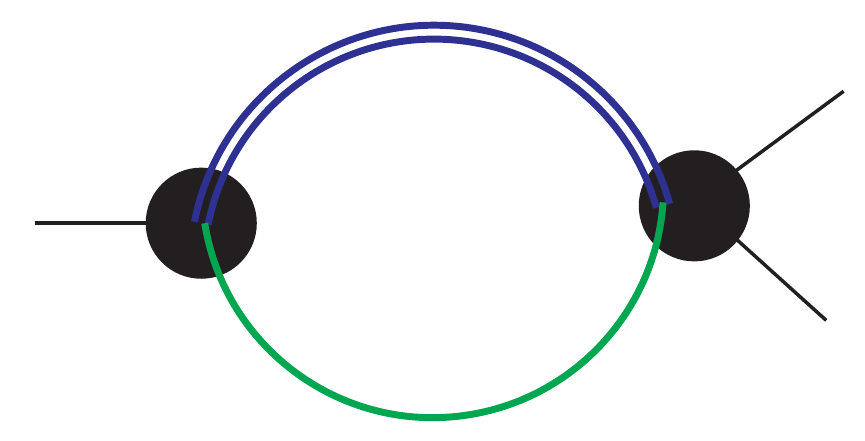}}\;}
\def\bubblegbbb{\;\raisebox{-5mm}{\includegraphics[width=7mm]{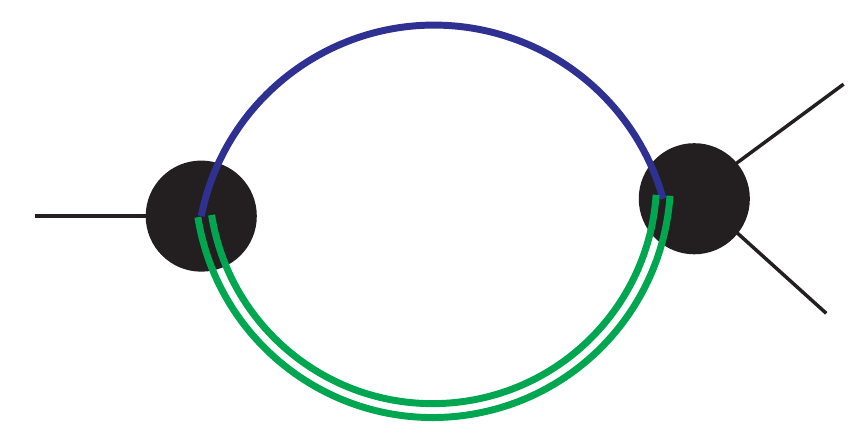}}\;}
\def\bubblegbg{\;\raisebox{-5mm}{\includegraphics[width=7mm]{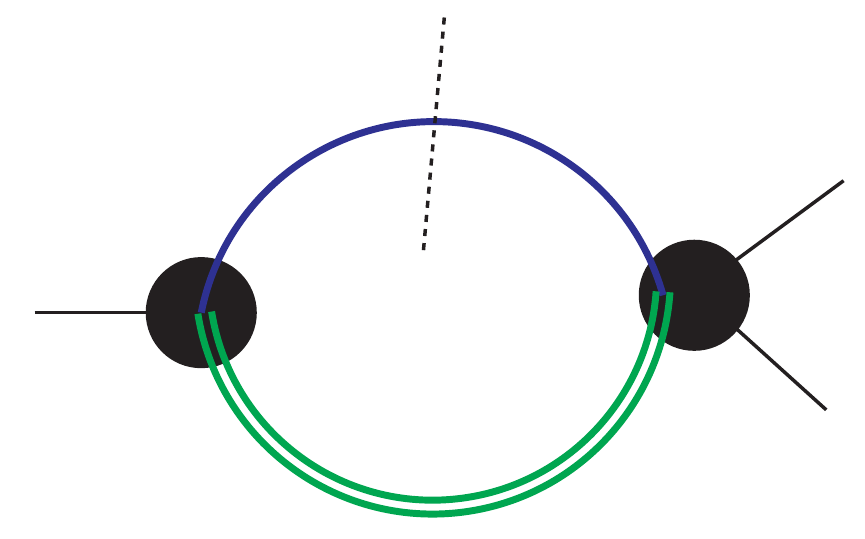}}\;}
\def\bubblegbb{\;\raisebox{-5mm}{\includegraphics[width=7mm]{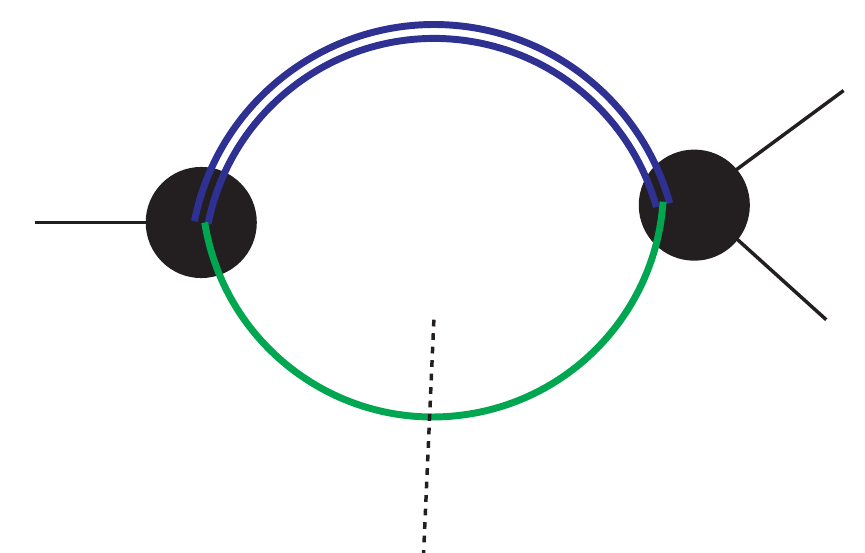}}\;}
\def\bubblegr{\;\raisebox{-4mm}{\includegraphics[width=7mm]{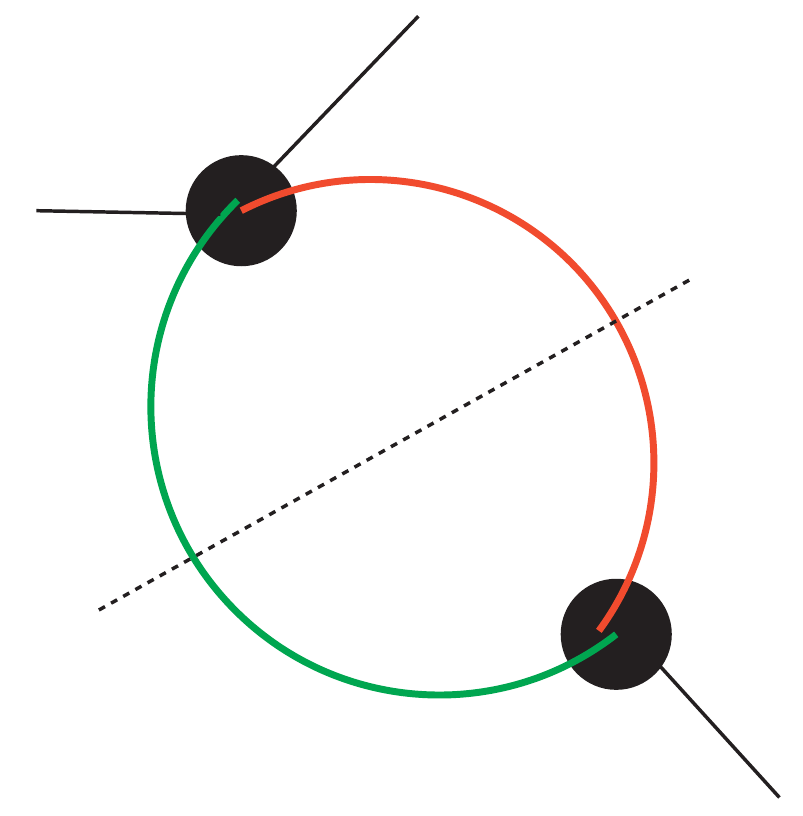}}\;}
\def\bubblegrgr{\;\raisebox{-4mm}{\includegraphics[width=7mm]{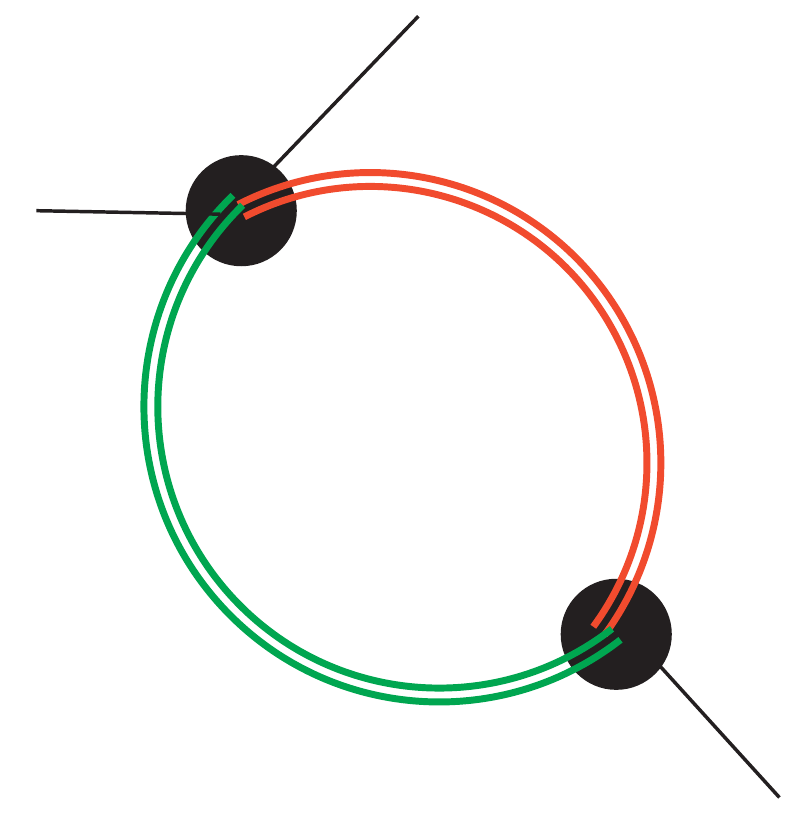}}\;}
\def\bubblegrgg{\;\raisebox{-4mm}{\includegraphics[width=7mm]{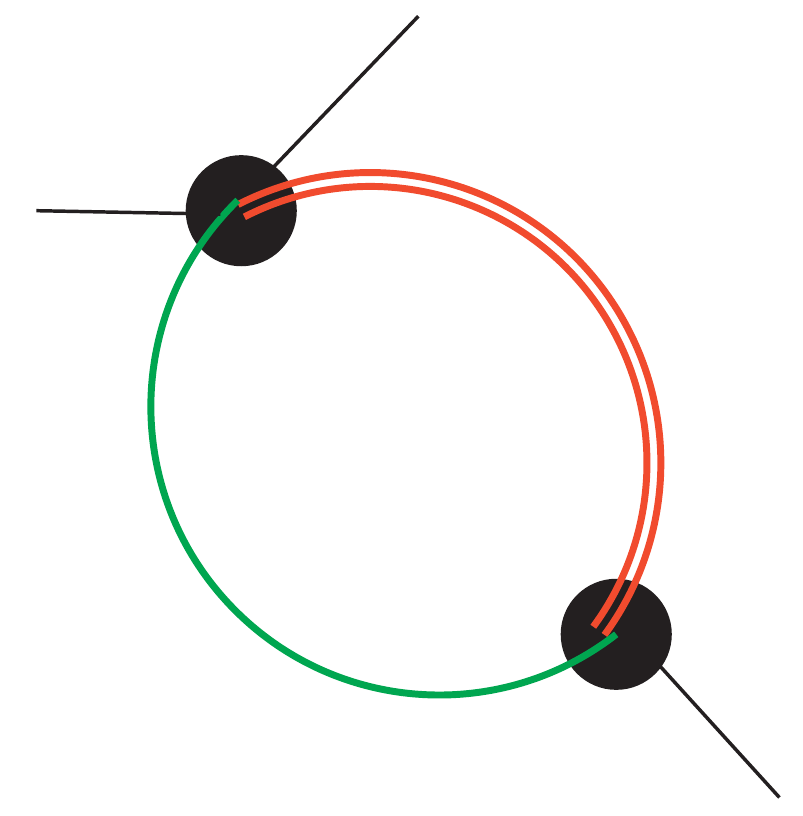}}\;}
\def\bubblegrg{\;\raisebox{-4mm}{\includegraphics[width=7mm]{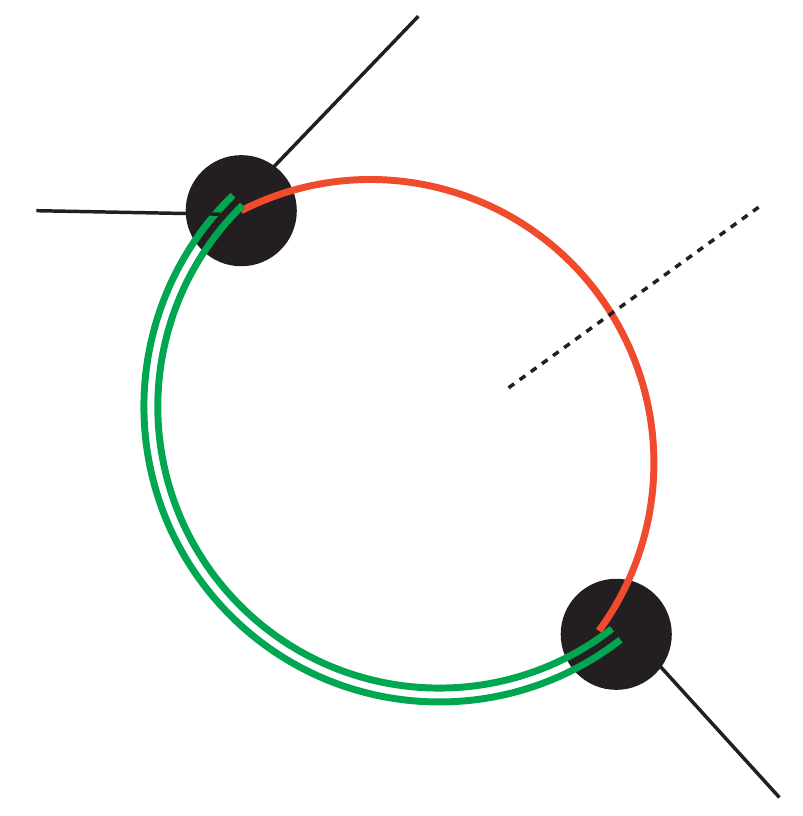}}\;}
\def\bubblegrgxr{\;\raisebox{-4mm}{\includegraphics[width=7mm]{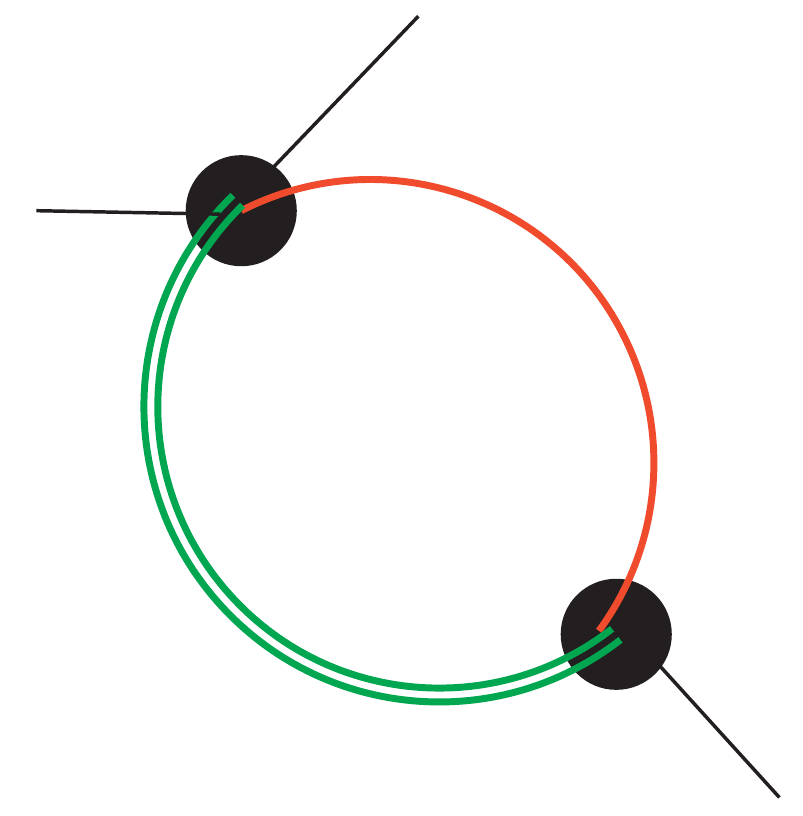}}\;}
\def\bubblegrr{\;\raisebox{-4mm}{\includegraphics[width=7mm]{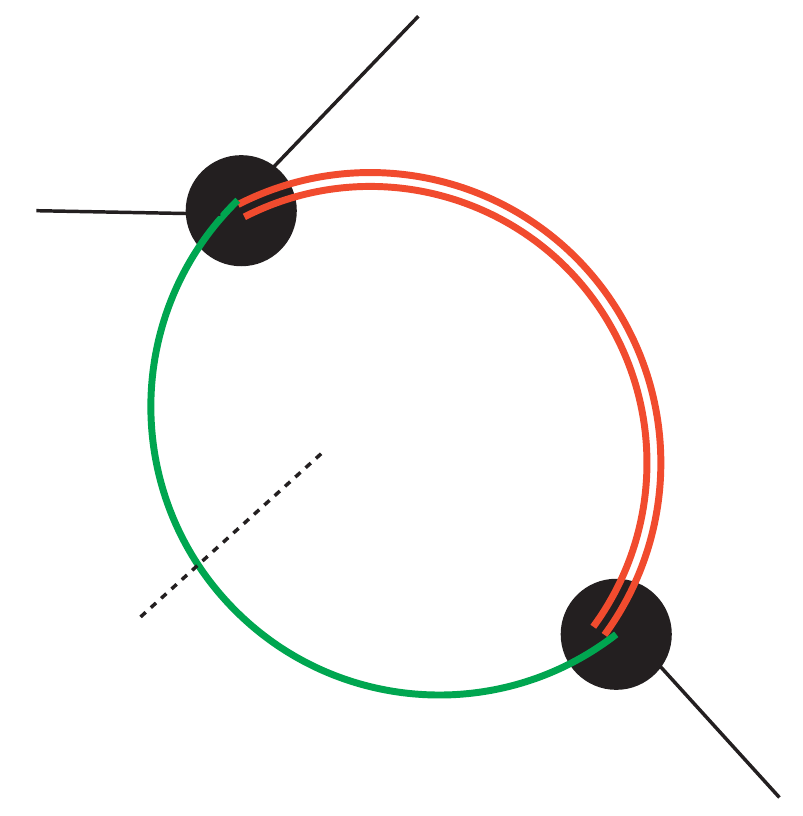}}\;}
\def\bubblebr{\;\raisebox{-4mm}{\includegraphics[width=7mm]{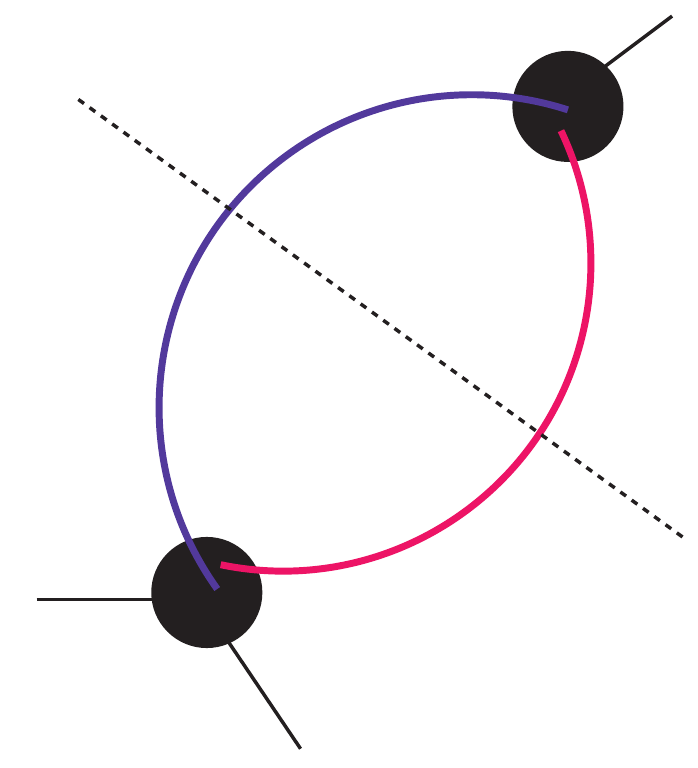}}\;}
\def\bubblebrbr{\;\raisebox{-4mm}{\includegraphics[width=7mm]{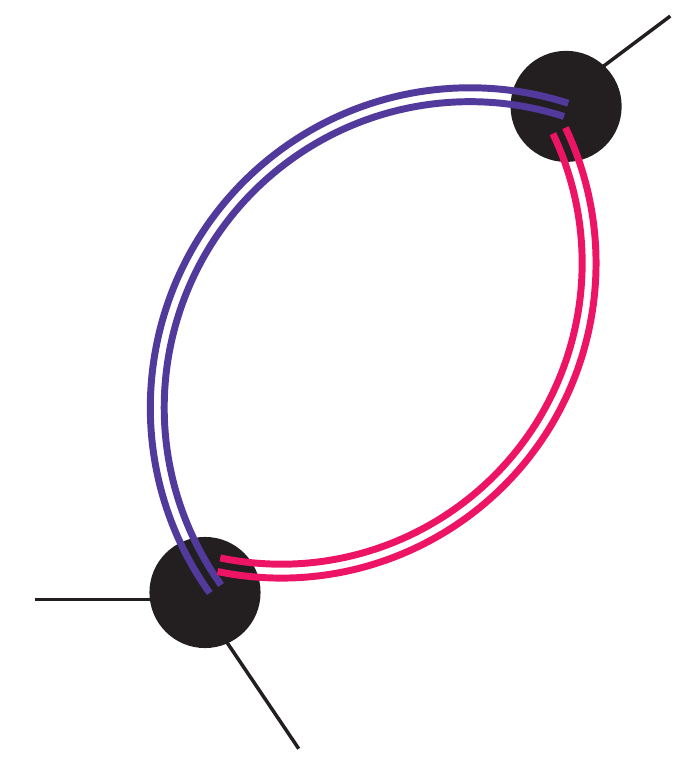}}\;}
\def\bubblebrbb{\;\raisebox{-4mm}{\includegraphics[width=7mm]{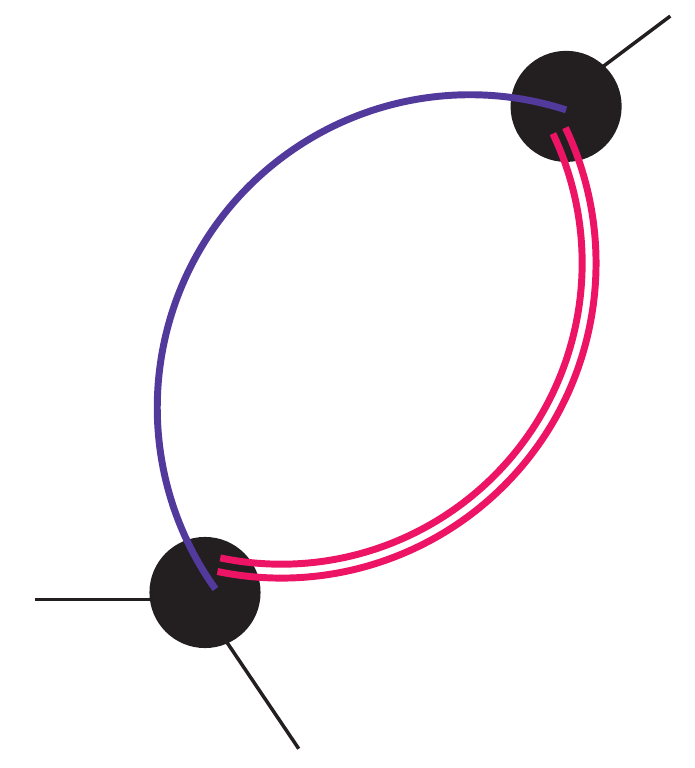}}\;}
\def\bubblebrb{\;\raisebox{-4mm}{\includegraphics[width=7mm]{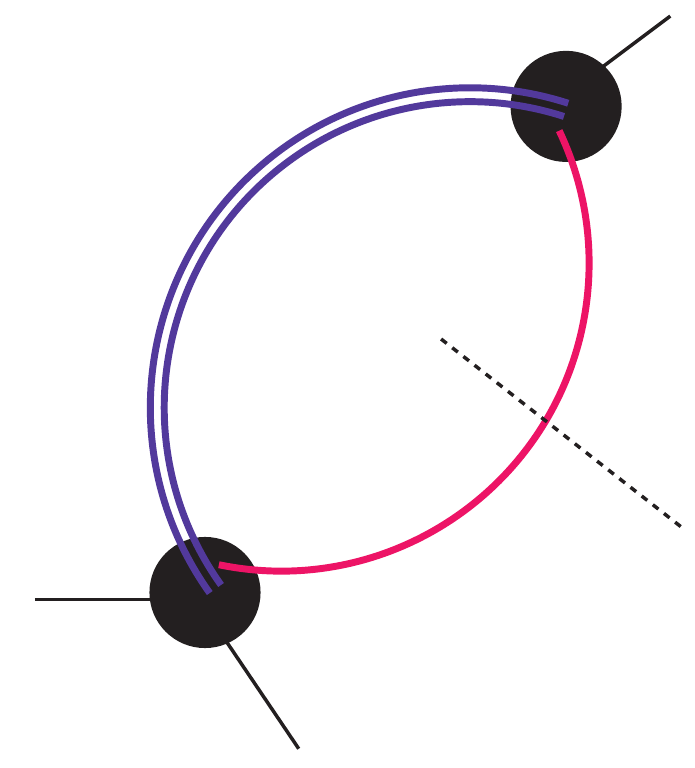}}\;}
\def\bubblebrr{\;\raisebox{-4mm}{\includegraphics[width=7mm]{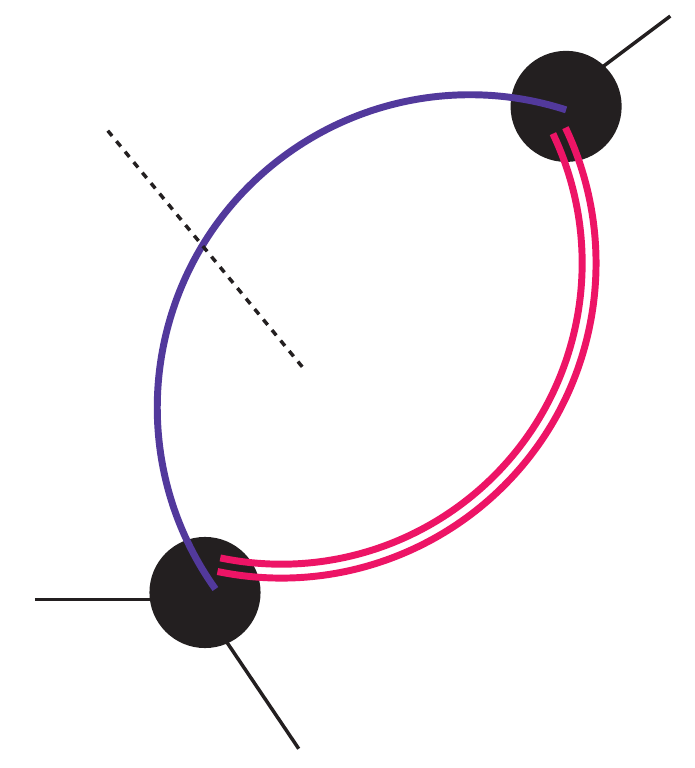}}\;}
\def\bubblebrrr{\;\raisebox{-4mm}{\includegraphics[width=7mm]{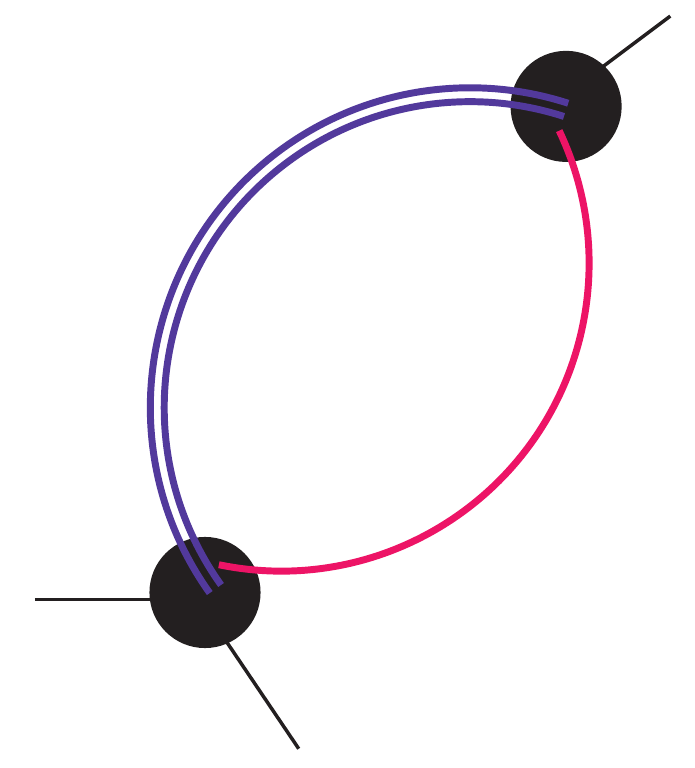}}\;}
\def\nnn{\;\raisebox{-3mm}{\includegraphics[width=8mm]{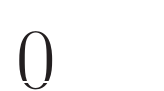}}\;}
\def\ooo{\;\raisebox{-3mm}{\includegraphics[width=8mm]{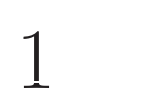}}\;}

\def\cycleswheel{\;\raisebox{-20mm}{\includegraphics[width=40mm]{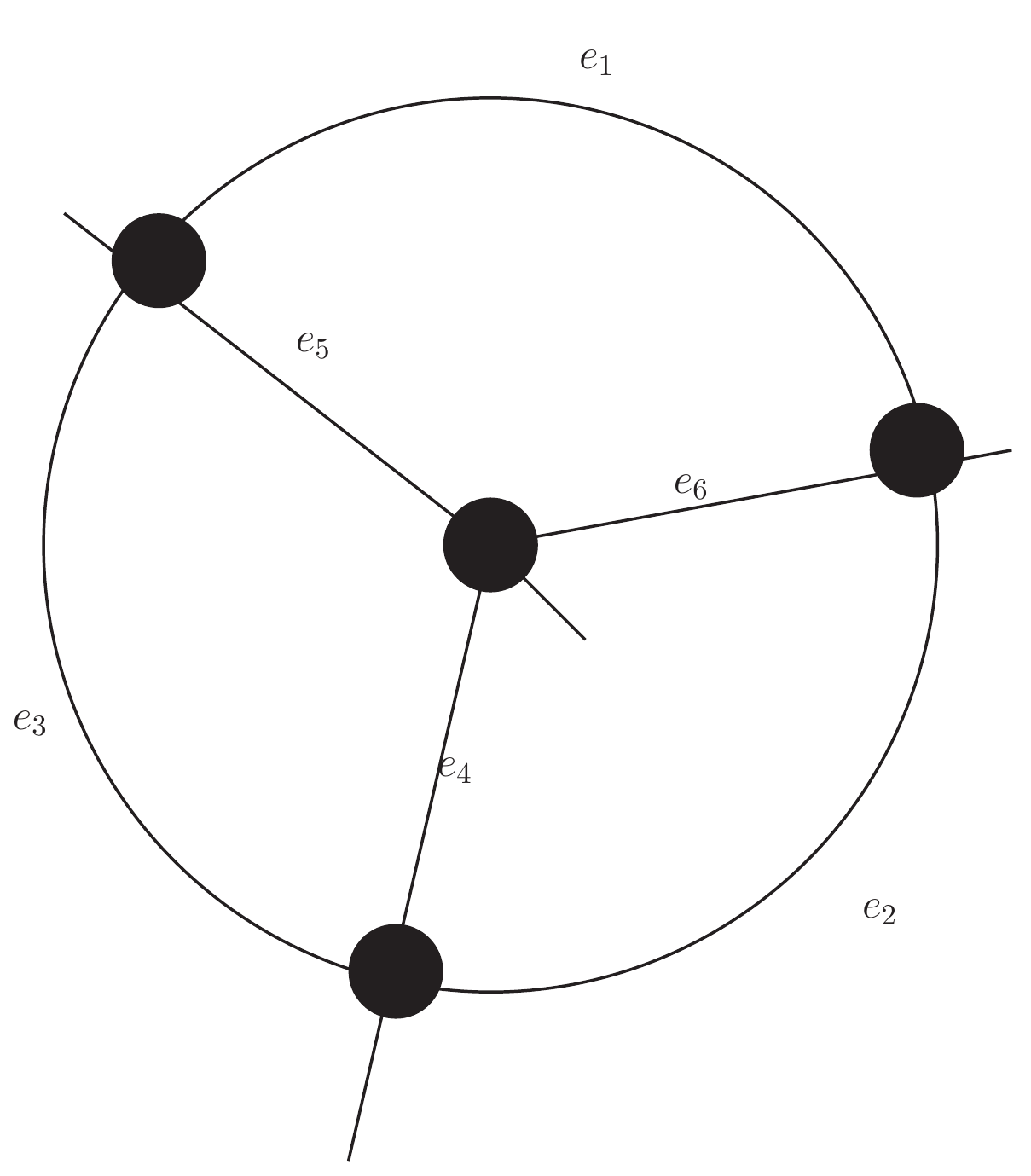}}\;}
\def\onedimcell{\;\raisebox{-4mm}{\includegraphics[width=100mm]{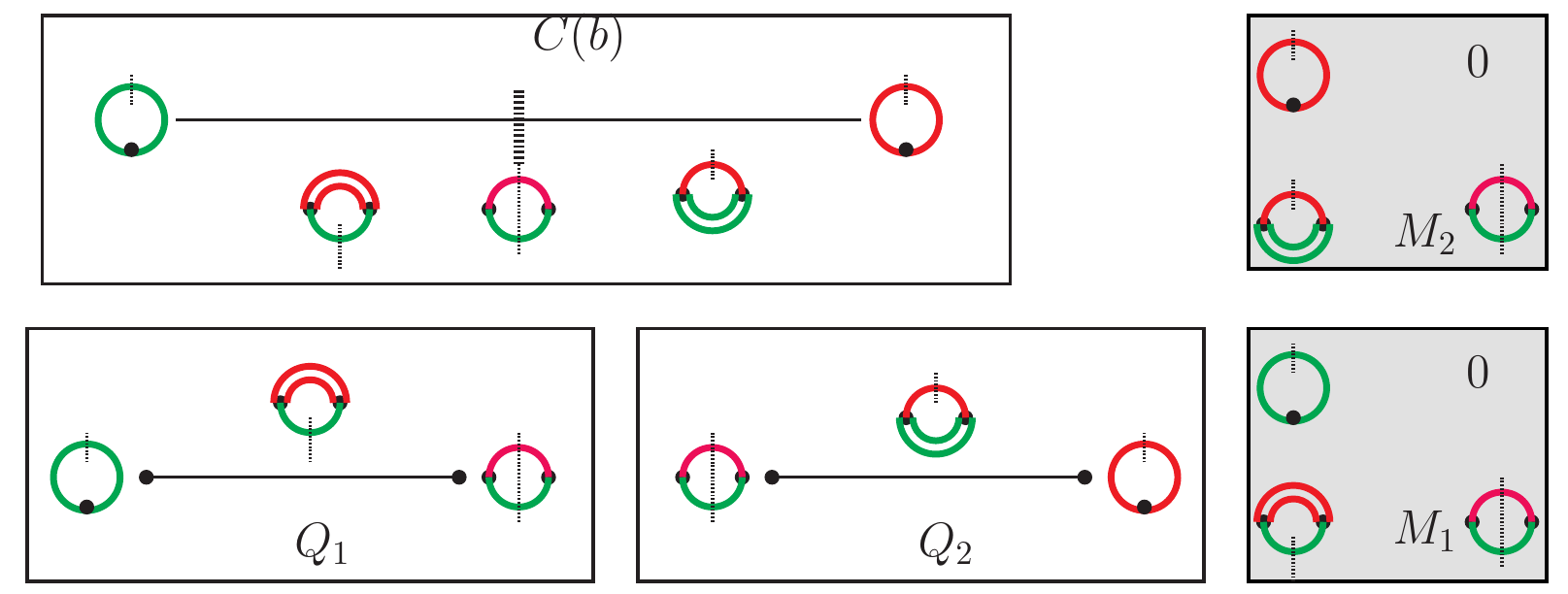}}\;}

\usepackage{float}
\floatstyle{boxed} 
\restylefloat{figure}

\begin{document}

\title*{Outer Space as a combinatorial backbone for Cutkosky rules and coactions}
\titlerunning{Outer Space and Coactions} 
\author{Dirk Kreimer}
\institute{Depts.\ of Mathematics and Physics \at  Humboldt Univ., 10099 Berlin, Germany, \email{kreimer@math.hu-berlin.de}
}
%
%
\maketitle

\abstract{We consider a coaction which exists for any bridge-free graph. It is based on the cubical chain complex associated to any such graph by considering two boundary operations: shrinking edges or removing them.
Only if the number of spanning trees of a graph $G$ equals its number of internal edges we find that the graphical coaction $\Delta^G$ constructed here agrees with the coaction $\Delta_{\mathsf{Inc}}$ proposed by Britto and collaborators.
The graphs for which this is the case are one-loop graphs or their duals, multi-edge banana graphs. They provide the only examples discussed by Britto and collaborators so far. We call such graphs simple graphs.
The Dunce's cap graph is the first non-simple graph. The number of its spanning trees (five) exceeds the number of its edges (four). 
We compare the two coactions which indeed do not agree and discuss this result.
We also point out that for kinematic renormalization schemes the coaction $\Delta^G$ simplifies.}

\section{Introduction}
The notion of a coaction has gained prominence recently in the context of amplitude computations in high energy physics \cite{Brittoetal}.

This is motivated by the appearance of multiple polylogarithms and their elliptic cousins in such computations \cite{BlKr,BlKerr,Weinz}. For such functions the existence of such a coaction is known. Indeed Francis Brown gave a masterful account of is appearance and conceptual role  \cite{BrownI,BrownII} in particular also with regards to the small graphs principle, see for example \cite{BrownIII,Taskupovic}.

For physicists it is second nature to regard any Feynman integral computation as a manipulation on Feynman graphs. 

One hence wishes to identify coactions in combinatorial manipulations of Feynman graphs which are  in accordance with their appearance in the study of such polylogarithms.

A possible  approach is based on reverse engineering by pulling back the coaction structure of the functions into which Feynman graphs evaluate 
to the graphs themselves. This is an approach successfully employed by Britto et.al.\ \cite{Brittoetal}
and they conjecture a graphical coaction which by construction is correct 
for the known graphs amenable to computation. 

These are the simple graphs alluded to in the abstract above plus a few non-simple graphs with kinematics chosen such that a large number of terms
in their conjecured graphical coactions is bound to vanish.  As a consequence  they evaluate to multiple polylogarithms (MPLs) and one is again in safe terrain and can again pull back to the coaction on MPLs.

Here we introduce a coaction which exists independent of any physics consideration as a purely mathematical construct.

Its existence follows from known studies in graph complexes and graph homology \cite{CullerV,Kontsevich,CoHaKaVo}. We derive it in all detail.

We then compare the constructions of Britto et.al.\ and ours and show that the two constructions agree on simple graphs and on graphs which evaluate to mere MPLs.

Next we  discuss  differences for generic kinematics for non simple graphs and argue why the suggestion  of Britto et.al.\ \cite{Brittoetal} for a graphical coaction is bound to fail.

We also discuss simplifications apparent in kinematical renormalization 
schemes and relate the Steinmann relations \cite{Steinmann} to the structure of the cubical chain complex.

\section{Incidence Hopf algebras for (lower) triangular matrices}
Let us first define incidence Hopf algebras following Schmitt \cite{SchmittI,SchmittII}. Apart from changes in notation this material is similar to the presentation in appendix C.3 of  Britto et.al.\ \cite{Brittoetal}.

We start from a (partially) ordered set $P$ with partial order $\leq$.

For $x,y\in P$, $x\leq y$, consider the interval
\[
[x,y]=\{z\in P|x\leq z\leq y\}.
\]  
Let $P_2$ be the $\mathbb{Q}$-algebra generated by such intervals
\cite{SchmittI,Brittoetal} through multiplication as a free product
by disjoint union of intervals.

It gives rise to an incidence bialgebra $I_P$ upon setting
\[
\Delta([x,y])=\sum_{x\leq z\leq y}[x,z]\otimes [z,y],
\]
for the co-product $\Delta$
and 
\be\label{counit} 
\hat{\One}([x,y])=\delta_{x,y},
\ee
for the co-unit $\hat{\One}$,
where 
\[
\delta_{x,y}=1,\,x=y,\, \delta_{x,y}=0,\, \text{else}.
\]
Note that $\Delta([x,x])=[x,x]\otimes [x,x]$ is group-like.

Following Schmitt we can turn this bialgebra $I_P$ into a Hopf algebra 
$I_{\tilde{P}}$ by augmenting $I_P$ by multiplicative inverses $[x,x]^{-1}$ for group-like
$[x,x]$, for all $x\in P$.

The antipode $S:I_{\tilde{P}}\to I_{\tilde{P}}$, $S\circ S=\mathrm{id}$
is defined by $S([x,x])=[x,x]^{-1}$ and
\[
S([x,y])=\sum_{x=z_0\leq z_1\leq\ldots\leq z_k=y}(-1)^k\frac{1}{[x,x]}
\prod_{i=1}^k \frac{[z_{i-1},z_i]}{[z_i,z_i]}.
\]
\subsection{Example: lower triangular matrices}
As an example consider lower triangular $n\times n$ matrices $M$, $M_{i,j}=0,\,j\gneq i$.
$P$ is provided by the first $n$  integers and the intervals $[ji]$, $1\leq j\leq i\leq n$, are represented as $M_{i,j}$. 

As $M_{i,j}\in H_{GF}$, a Hopf algebra \cite{MarkoDirk},  the algebra structure of $P_2$ agrees with the algebra structure of $H_{GF}$ and is a free commutative algebra.
We have $\Delta(M_{i,j} M_{l,s})=\Delta(M_{i,j})\Delta(M_{l,s})$.

$I_P\equiv I_M$ gets a different bialgebra structure though.  
Instead of using the coproduct $\Delta_{GF}$ of $H_{GF}$ the coproduct is
\[
\Delta M_{i,j}=\sum_{k=j}^i M_{k,j}\otimes M_{i,k}. 
\]
Coassociativity of this map is obvious.
\beas
(\Delta\otimes\mathrm{id})\Delta(M_{j,k}) & = & \sum_{h,i}M_{h,k}\otimes M_{i,h}\otimes M_{j,i},\\
(\mathrm{id}\otimes\Delta)\Delta(M_{j,k}) & = & \sum_{h,i}M_{i,k}\otimes M_{h,i}\otimes M_{j,h},
\eeas
and the two expressions on the right obviously agree.

Consider the $\mathbb{
Q}$-vectorspace $V_1$ generated by elements $M_{i,1}$, $i\gneq 1$. Let $\rho_\Delta:
V_1\to V_1\otimes I_{\tilde{P}}$, 
\[
\rho_{\Delta}(M_{i,1})=\sum_{k=2}^{i}
M_{k,1}\otimes M_{i,k},
\]
be the restriction of $\Delta$ to $V_1$.

Then coassociativity of $\Delta$ delivers
\[
(\mathrm{id}\otimes\Delta)\rho_\Delta=(\rho_\Delta\otimes\mathrm{id})\rho_\Delta,
\]
and we also have by Eq.(\ref{counit}) that $(\mathrm{id}\otimes \hat{\One})\rho_\Delta=\mathrm{id}$.
We conclude
\begin{prop}
$\rho_\Delta$ is a coaction on $V_1$.
\end{prop}
Note that we get such a Hopf algebra and coaction for any chosen lower triangular matrix $M$. We write $\Delta^M$ whenever necessary.

It is useful to define matrices $M_\One$
\[
(M_\One)_{i,j}=M_{i,j}/M_{i,i},
\]
which has unit entries along the diagonal and the diagonal matrix $M_{\mathsf{D}}$
\[
(M_{\mathsf{D}})_{i,j}=0,\,i\not= j,(M_{\mathsf{D}})_{i,i}=M_{i,i}. 
\]
Then,
\be\label{factimag}
(M_{\mathsf{D}})\times M_\One=M.
\ee

For lower triangular matrices there are two maps which are natural to consider: shifting to the row above or to the column to the right.

So consider the map 
\[
m_r:\,M_{i,j}\to M_{i-1,j},
\]
where we set $M_{0,j}=0$
and the map 
\[
m_c:\,M_{i,j}\to M_{i,j+1},
\]
where we set $M_{n,n+1}=0$.
\begin{prop}
We have
\be 
(\mathrm{id}\otimes m_r)\circ \Delta=\Delta\circ m_r,
\ee
and
\be 
(m_c\otimes \mathrm{id})\circ \Delta=\Delta\circ m_c.
\ee
\end{prop}
\begin{proof}
$m_r$: $\Delta$ maps entries from the $j$-th row to entries in the $j$-th row
on the rhs of the tensorproduct, and $m_r$ shifts $j\to j-1$ on both sides of the equation.\\
$m_c$: $\Delta$ maps entries from the $k$-th column to entries in the $k$-th column 
on the lhs of the tensorproduct, and $m_c$ shifts $k\to k+1$ on both sides of the equation.
\end{proof}
\begin{ex}
Let us consider an example for all the above:
\be\label{matrixexa}
M=\left(
\begin{array}{ccccccc}
1 & & 0 & & 0  & & 0\\
a & & b & & 0 & & 0\\
c & & d & & e & & 0\\
f & & g & & h & & j\\
\end{array}\right).
\ee
Then,
\be
M_\One=\left(
\begin{array}{ccccccc}
1 & & 0 & & 0  & & 0\\
a/b & & 1 & & 0 & & 0\\
c/e & & d/e & & 1 & & 0\\
f/j & & g/j & & h/j & & 1\\
\end{array}\right),
\ee
and
\be
M_{\mathsf{D}}=\left(
\begin{array}{ccccccc}
1 & & 0 & & 0  & & 0\\
0 & & b & & 0 & & 0\\
0 & & 0 & & e & & 0\\
0 & & 0 & & 0 & & j\\
\end{array}\right).
\ee
We have 
\[
\rho_\Delta(f)=f\otimes j+ c\otimes h+a\otimes g,
\]
when we regard $f$ as an element of the $\mathbb{Q}$-vectorspace $V_1$
spanned by $a,c,f$. Also $\Delta(g)=g\otimes j+d\otimes h+b\otimes g$,
where $g\in I_M\setminus V_1$, and $I_M$ is the $\mathbb{Q}$-vectorspace spanned by 
the entries of $M$.
 
Furthermore
\[
\Delta(m_r(g))=\Delta(d)=d\otimes e + b\otimes d=(\mathrm{id}\otimes m_r)
(g\otimes j+d\otimes h+b\otimes g),
\] 
as $m_r(j)=0$,
and
\[
\Delta(m_c(g))=\Delta(h)=h\otimes j + e\otimes h=(m_c\otimes \mathrm{id})
(g\otimes j+d\otimes h+b\otimes g),
\] 
as $m_c(b)=0$.
\end{ex}
Let us introduce some more terminology.
For any lower triangular $n\times n$ matrix $M$ let us call the entries $M_{i,j}$ the Galois correspondents of $M$, $M_{i,j}\in \mathsf{Gal}(M)$.
We regard $\Delta$ as a map  
\[
\mathsf{Gal}(M)\to \mathsf{Gal}(M)\otimes \mathsf{Gal}(M).
\]

For several say $k$ such matrices $M_i$, $1\leq i\leq k$,  each of them giving rise to a coproduct and coaction $\Delta^i\equiv \Delta^{M_i}$ we associate the set $\mathsf{Gal}_k:=\cup_{i=1}^k \mathsf{Gal}(M_i)$.
The union is not a disjoint union as a single Galois correspondent 
can be contained in various sets $\mathsf{Gal}(M_i)$ of such correspondents simultaneously.

We then define for all $x\in \mathsf{Gal}_k$,
\[
\Delta(x)=\sum_{j=1}^k\Delta^j(x),
\]
where we set $\Delta^j(x)=0$ $\forall x\not\in \mathsf{Gal}(M_j)$.

In fact there is a matrix $M$ which we can assign to
$\mathsf{Gal}_k$. Of particular interest to us is the case where the entries in the upper left and lower right corner are all equal:  $(M_i)_{1,1}=(M_j)_{1,1}$
and $(M_i)_{n,n}=(M_j)_{n,n}$, $\forall i,j$. 

The generic construction  is an obvious iteration of the following example on two matrices.
\begin{ex}
Assume $M_1,M_2$ are lower $k\times k$ square matrices.

So the $M_i^B$ below are lower $(k-2)\times (k-2)$ square matrices, while
the $M_i^C$ are $(k-2)\times 1$ column matrices, the $M_i^R$ are $1\times (k-2)$ row matrices, $i\in \{1,2\}$. 
\[
M_1=\left(
\begin{array}{c|ccccccc|c}
1 & & 0 & & \sim  & & 0 & & 0\\
\hline
. & & . & & . & & . & & 0\\
M_1^C & & . & & M_1^B & & . & & \wr\\
. & & . & & . & & . & & 0\\
\hline
g & & . & & M_1^R & & . & & c\\
\end{array}\right),
\]
\[
M_2=\left(
\begin{array}{c|ccccccc|c}
1 & & 0 & & \sim  & & 0 & & 0\\
\hline
. & & . & & . & & . & & 0\\
M_2^C & & . & & M_2^B & & . & & \wr\\
. & & . & & . & & . & & 0\\
\hline
h & & . & & M_2^R & & . & & c\\
\end{array}\right),
\]
Then,
\[
M=\left(
\begin{array}{c|cccccc|ccccccc|c}
1 & & 0 & & \sim  & & 0 & & 0 & & \sim & & 0 & & 0\\
\hline
. & & . & & . & & . & & \sim & & \sim & & \sim & & 0\\
M_1^C & & . & & M_1^B & & . & & \sim & & 0 & & \sim & & \wr\\
. & & . & & . & & . & & \sim & & \sim & & \sim & & 0\\
\hline
. & & \sim & & \sim & & \sim & & . & & . & & . & & 0\\
M_2^C & & \sim & & 0 & & \sim & & . & & M_2^B & & . & & \wr\\
. & & \sim & & \sim & & \sim & & . & & . & & . & & 0\\
\hline
g+h & & . & & M_1^R & & . & & . & & M_2^R & & . & & c\\
\end{array}\right),
\]
And indeed immediately checks
\[
\Delta^M=
\Delta^{M_1}+\Delta^{M_2}.
\]
\end{ex}

Furthermore note that for any entry $M_{ij}$ in a $n\times n$ matrix $M$ there exists a lower triangular $(i-j+1)\times (i-j+1)$ matrix 
$M^{ij}$,
with
\[
\Delta^M(x)=\Delta^{M^{ij}}(x),\,\forall x\in \mathsf{Gal}(M^{ij}).
\]
Here the lowest leftmost entry of the matrix $M^{ij}$ is $M_{i,j}$.
\section{Lower triangular matrices from the cubical chain complex}
Lower triangular matrices derived from the cubical chain complex played a prominent role already in \cite{BlKrCut}.  We refine their construction here to derive a graphical coaction.
 
A most prominent role in the study of Feynman graphs $G$ is played by their $|G|$ independent loops. They provide the basis for the subsequent loop integrations of Feynman integrals. 

Assume given a bridgefree graph $G$ together with a spanning tree $T$ for it constituting a pair $(G,T)\equiv G_T\in H_{GF}$.

We can route all external momentum flow through edges $e\in E_T$ of the spanning tree. The remaining $|G|=e_G-e_T$ edges $e_i$ generate a basis $L_{G_T}$ for the cycle space of $G$,
\[
L_{G_T}=\{\cup_{i=1}^{|G|}l_i\},
\]
where each $l_i$ is a cycle of edges given by a pair $(e_i,p_i)$ where $e_i\not\in E_T$ and $p_i\subseteq E_T$ is the unique path of edges connecting source and target of $e_i$. 

The Feynman integral is independent of the choice of $T$ as long as $T$ is a spanning tree of $G$, $T\in\mathcal{T}(G)$.
\begin{ex}
Here is the wheel with three spokes $w_3$:
\[
w_3=\cycleswheel
\]
Its cycles are 
\beas
\{(e_1,e_2,e_3),(e_1,e_5,e_6),(e_2,e_4,e_6),(e_3,e_5,e_4),(e_2,e_3,e_5,e_6),(e_6,e_4,e_3,e_1),\\(e_4,e_5,e_1,e_2)\}.
\eeas
 It has $16=|\mathcal{T}(w_3)|$ spanning trees $T$ on three edges. Choose $T=(e_1,e_6,e_4)$.
$e_2,e_3,e_5$ generate a basis for the cycles $l_1:=(e_2,p_2)$, $l_2:=(e_3,p_3)$ and $l_3:=(e_5,p_5)$, and $p_2=(e_4,e_6)$, $p_3=(e_4,e_6,e_1)$ and $p_5=(e_1,e_6)$.
\end{ex}
 
This setup suggests to study Culler--Vogtmann {\em Outer Space} \cite{CullerV}.
It assigns a $k$-dimensional cell $C(G)$ to any graph $G$  on
$k+1$ edges and the Feynman integral becomes an integral over the volume over this cell. 

This is evident in parametric space where we can identify the edge length $A_e$ of an edge $e$ with the parametric variable. 

The renormalized Feynman form (see \cite{BrownKreimer} for notation
and for other than log-divergent singularities in renormalization)
\[
\mathsf{Int}_R (G)(q,p)=\sum_{F}(-1)^{|F|}\frac{\ln\left(\frac{\Phi(G/F)\psi(F)+\Phi_0(F)\psi(G/F)}{\Phi_0(G/F)\psi(F)+\Phi_0(F)\psi(G/F)}\right)}{\psi^2_G}\Omega_G
\]
as provided by the Symanzik polynomials $\Phi,\psi$ gives the volume form for $C(G)$. The sum is over the forests of $G$ as demanded by renormalization  and we get
\[
\Phi_R(G)(q)=\int_{C(G)} \mathsf{Int}_R(G).
\]
Here, $q\in \mathbf{Q}_G$ is a vector spanned by Lorentz invariants in $\mathbf{Q}_G$  as provided by the external momentum vectors of $G$ and $p\in\mathbb{P}_G$ is a point in the projective space spanned by positive real edge variables.

A sum over all vacuum Feynman graphs then integrates Outer Space (OS). 
A sum over all bridgefree graphs with $n$ loops and  a given number $s$  of external legs  sums the corresponding classes of graphs $X_{n,s}$  with $s$ fixed marked points  and $n$ loops \cite{CoHaKaVo}. 
 
The codimension $k$ boundaries of cells $C(G)$ are cells themselves which are assigned to reduced graphs where $k$ edges of $T$ and therefore of $G$ shrink.

All codimension one boundaries appear for example as
\[
C(G)\to C(G/e),\,e\in E_G.
\]
The lowest dimensional boundary is of codimension $e_T$ and the graph assigned to it is the rose $G/E_T$.

When studying $C(G)$ and the cells apparent as codimension $k$ boundaries a prominent role its played by the barycenters of all these cells.
\begin{ex}
Consider the simplest cell $C(b)$, a one-dimensional line for the bubble graph $b$ on two edges with different masses (green or red lines)  
with two 1-edge spanning trees indicated by double edges.
\[
\onedimcell
\]
The codimension one ends are zero-dimensional cells to which a tadpole graph is associated. At the barycentric middle of the cell we have the graph with its two vertices as spanning forest, and the two internal edges on-shell. The barycentric middle is determined by 
\[
A_rm_r=A_gm_g
\]
as a point $p\in\mathbb{P}_b\equiv \mathbb{P}^1$ which determines the ratio $p_b=A_r/A_g=m_g/m_r$, and there is $q_b\in \mathbf{Q}_b\sim\mathbb{R}$ (generated by $q^2$ with $q_b:\,q^2=(m_1+m_2)^2$),
so that $(q_b,p_b)$ determines a threshold divisor in $\mathsf{Int}_R(b)(q,p)$.

The cell $C(b)$ gives by the two spanning trees rise to two one-dimensional unit cubes $Q_1\sim [0,1]$, $Q_2\sim [0,1]$. Approaching zero, the graph shrinks to a tadpole, approaching one it approaches the leading singularity 
\[
\mathsf{Int}_R(b)(q,p),
\]
corresponding to the barycenter and the evaluation of $\Phi_R(b)$ at  $q^2=(m_1+m_2)^2$. As each cube is one-dimensional it gives rise to a single 
lower triangular Hodge matrix as indicated.
\end{ex}

For each cell $C(G)$ with associated (reduced) graph $g$ the barycenter corresponds to the leading singular graph $g_F$, where $F=\cup_{v\in V_g}$,
so that all internal edges of $g$ are on the mass-shell.

Such barycenters  define paths from the barycenter of $C(G)$ 
through barycenters of lower and lower dimensional hypersurfaces until we reach the barycenter of  codimension $e_T$ cells. 

The collection of all these paths defines the {\em spine of Outer Space} as a deformation retract of OS \cite{CullerV}.

The barycenters  then provide the coordinates in parametric space of the threshold divisors which generate monodromy of Feynman amplitudes. 

Physicsal thresholds are determined by solving a variational problem \cite{BlKrCut} determining  the minimal kinematical configuration so as to make the discriminant of the second Symanzik polynomial $\Phi$  vanish for the associated leading singular graph.

The spine determines a set of $|\mathcal{T}(G)|$ $e_T$-dimensional cubes
and with it $|\mathcal{T}(G)|\times e_T!$ paths from the midpoint of $C(G)$ to the rose $G/E_T$. Each such path defines a lower triangular matrix correponding 
to the $e_T!$ simplices into which a cube decomposes.

We will thus now turn to the cubical chain complex for Feynman graphs \cite{MarkoDirk,BlKrCut} where a pair $G_T$  of a graph $G$ and a spanning tree $T$ for it gives rise
to a $e_T$-cube.

Any $e_T$-cube gives rise to a natural cell decomposition into
$e_T!$ simplices and therefore generates $e_T!$ lower triangular matrices
corresponding to the $e_T!$ possible orderings of the edges.
Fig.(\ref{trianglecanvas}) is instructive.
\begin{figure}[H]
\includegraphics[height=66mm]{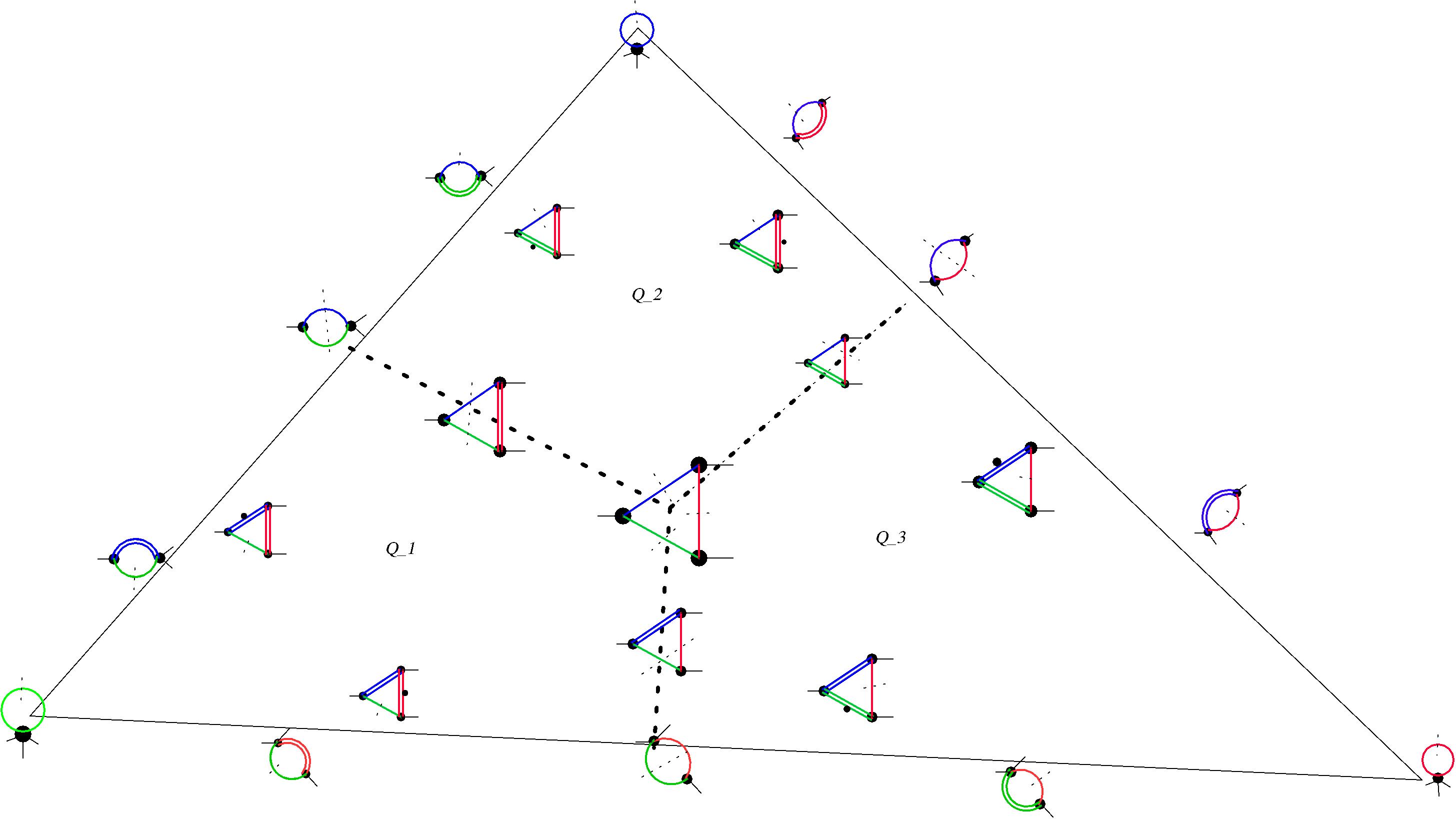}
\caption{A two dimensional cell $C(t)$  (a 2-simplex, so itself a triangle) in outer space for the triangle graph $t$ on three different masses (indicated by colored edges). On-shell edges are thin and marked by a hashed line, off-shell edges are double lines, a dot orders the two-edge spanning trees with the dotted edge the longer one. For this simple graph the spine gives a simplicial decomposition of $C(t)$ into three open 2-cubes $Q_1,Q_2,Q_3$.}
\label{trianglecanvas}
\end{figure}
Note that the cell $C(t)$ as well as the cubes $Q_i$ are two dimensional and in fact the cell $C(t)$ can be dissected in $|\mathcal{T}(t)|=3$ open cubes so that $C(t)=\overline{\amalg_i Q_i}$ is the completion of their union. This is in fact typical for one-loop graphs:
\[
C(G)=\bar{X},\,X=\amalg_{j=1}^{|\mathcal{T}(G)|} Q_j,\,1\leq j\leq |\mathcal{T}(G)|
\]
and $\bar{X}\setminus X$ is the spine of $G$. This is very different for generic graphs where $\mathsf{dim}(C(G))=\mathsf{dim}(Q_j)=-1+|G|$ \cite{MarkoDirk}.

Let us study one cube say for the spanning tree on blue and red edges, so the cube $Q_1$ containing $\tadpoleg$.

It provides nine cells: a two-dimensional square, four one-dimensional edges, and four zero-dimensional corners.
\[
Q^{(1)}=
\begin{array}{cc|cccccc|ccc}
\bubblegb & & & & \triangler  & & & & \triangle & \\
\hline
& & & & & & & & & \\
& & & & & & & & & \\
\bubblegbb & & & & \trianglebr & & & & \triangleb & \\
& & & & & & & & & \\
& & & & & & & & & \\
\hline
\tadpoleg & & & &  \bubblegrr & & & & \bubblegr &\\
\end{array}
\]
The three graphs in the anti-diagonal from the lowest left corner (the origin of the cube) to the upper right corner appear in both Hodge matrices $M^{(1)}$ and $M^{(2)}$ defined below and associated to this cube. 

In such $n$-cubes, the $n!$  paths from the rose (the origin of the cube) to the leading singular graph (diagonally opposed for the main diagonal) share their origin and their endpoint and do not intersect otherwise. This reflects the Steinmann relations \cite{Steinmann} as threshold divisors do not overlap between graphs and sectors assigned to different paths. An example:
\begin{ex}
We list all six matrices $M^{(i)}$, $1\leq i\leq 6$ for the triangle graph.
{
\[
M^{(1)}=\left(
\begin{array}{cc||cc|cc|c}
\ooo & & \nnn & & \nnn  & & \nnn\\
\hline
\hline
\tadpoleggg & & \tadpoleg & & \nnn  & & \nnn\\
\hline
\bubblegbgg & & \bubblegbb & & \bubblegb & & \nnn\\
\hline
\trianglebrbg & & \trianglebrb & & \triangler & & \triangle\\
\end{array}\right),
\]
for sectors
\[
\left(
\begin{array}{cc|cc}
1 & & \emptyset\\
\hline
\hline
\tadpoleggg & & A_g>0\\
\hline
\bubblegbgg & & A_g>A_b>0\\
\hline
\trianglebrbg & & A_g>A_b>A_r>0\\
\end{array}\right),
\]
\[
M^{(2)}=\left(
\begin{array}{cc||cc|cc|c}
\ooo & & \nnn & & \nnn  & & \nnn\\
\hline
\hline
\tadpoleggg & & \tadpoleg & & \nnn  & & \nnn\\
\hline
\bubblegrgg & & \bubblegrr & & \bubblegr & & \nnn\\
\hline
\trianglebrrg & & \trianglebrr & & \triangleb & & \triangle\\
\end{array}\right),
\]
for sectors
\[
\left(
\begin{array}{cc|cc}
1 & & \emptyset\\
\hline
\hline
\tadpoleggg & & A_g>0\\
\hline
\bubblegrgg & & A_g>A_r>0\\
\hline
\trianglebrrg & & A_g>A_r>A_b>0\\
\end{array}\right).
\]
As $A_r>A_b$ and $A_b>A_r$ do not coexist, both sectors give monodromy in different 
regions of $\mathbf{Q}_t$.
\[
M^{(3)}=\left(
\begin{array}{cc||cc|cc|c}
\ooo & & \nnn & & \nnn  & & \nnn\\
\hline
\hline
\tadpolerrr & & \tadpoler & & \nnn  & & \nnn\\
\hline
\bubblegrgxr & & \bubblegrg & & \bubblegr & & \nnn\\
\hline
\trianglegbgr & & \trianglegbg & & \triangleb & & \triangle\\
\end{array}\right),
\]
for sectors
\[
\left(
\begin{array}{cc|cc}
1 & & \emptyset\\
\hline
\hline
\tadpolerrr & & A_r>0\\
\hline
\bubblegrgxr & & A_r>A_g>0\\
\hline
\trianglegbgr & & A_r>A_g>A_b>0\\
\end{array}\right),
\]
\[
M^{(4)}=\left(
\begin{array}{cc||cc|cc|c}
\ooo & & \nnn & & \nnn  & & \nnn\\
\hline
\hline
\tadpolerrr & & \tadpoler & & \nnn  & & \nnn\\
\hline
\bubblebrrr & & \bubblebrb & & \bubblebr & & \nnn\\
\hline
\trianglegbbr & & \trianglegbb & & \triangleg & & \triangle\\
\end{array}\right),
\]
for sectors
\[
\left(
\begin{array}{cc|cc}
1 & & \emptyset\\
\hline
\hline
\tadpolerrr & & A_r>0\\
\hline
\bubblebrrr & & A_r>A_b>0\\
\hline
\trianglegbbr & & A_r>A_b>A_g>0\\
\end{array}\right),
\]
\[
M^{(5)}=\left(
\begin{array}{cc||cc|cc|c}
\ooo & & \nnn & & \nnn  & & \nnn\\
\hline
\hline
\tadpolebbb & & \tadpoleb & & \nnn  & & \nnn\\
\hline
\bubblebrbb & & \bubblebrr & & \bubblebr & & \nnn\\
\hline
\trianglergrb & & \trianglergr & & \triangleg & & \triangle\\
\end{array}\right),
\]
for sectors
\[
\left(
\begin{array}{cc|cc}
1 & & \emptyset\\
\hline
\hline
\tadpolebbb & & A_b>0\\
\hline
\bubblebrbb & & A_b>A_r>0\\
\hline
\trianglergrb & & A_b>A_r>A_g>0\\
\end{array}\right),
\]
\[
M^{(6)}=\left(
\begin{array}{cc||cc|cc|c}
\ooo & & \nnn & & \nnn  & & \nnn\\
\hline
\hline
\tadpolebbb & & \tadpoleb & & \nnn  & & \nnn\\
\hline
\bubblegbbb & & \bubblegbg & & \bubblegb & & \nnn\\
\hline
\trianglerggb & & \trianglergg & & \triangler & & \triangle\\
\end{array}\right),
\]
for sectors
\[
\left(
\begin{array}{cc|cc}
1 & & \emptyset\\
\hline
\hline
\tadpolebbb & & A_b>0\\
\hline
\bubblegbbb & & A_b>A_g>0\\
\hline
\trianglerggb & & A_b>A_g>A_r>0\\
\end{array}\right).
\]
}
Note that each such Hodge matrix is well-defined in its sector.

We define 
\[
\Delta^M=\sum_{j=1}^6 \Delta^{M^{(j)}},
\]
where we find $M$ as
\[
\left(
\begin{array}{cc||cc|cc|cc|cc|cc|cc|c}
\ooo & & \nnn & & \nnn  & & \nnn & & \nnn & & \nnn & & \nnn & & \nnn\\
\hline
\hline
\tadpolegg & & \tadpoleg & & \nnn  & & \nnn  & & \nnn  & & \nnn  & & \nnn  & & \nnn\\
\hline
\tadpolerr & & \nnn & & \tadpoler & & \nnn & & \nnn & & \nnn & & \nnn & & \nnn\\
\hline
\tadpolebb & & \nnn & & \nnn & & \tadpoleb & & \nnn & & \nnn & & \nnn & & \nnn\\
\hline
\bubblegrgr & & \bubblegrr & & \bubblegrg  & & \nnn & & \bubblegr & & \nnn & & \nnn & & \nnn\\
\hline
\bubblebrbr & & \nnn & & \bubblebrb & & \bubblebrr & & \nnn & & \bubblebr & & \nnn & & \nnn\\
\hline
\bubblegbgb & & \bubblegbb & & \nnn & & \bubblegbg & & \nnn & & \nnn & & \bubblegb & & \nnn\\
\hline
\trianglebrg & & \trianglebr & & \trianglegb & & \trianglerg & & \triangleb & & \triangleg & & \triangler & & \triangle\\
\end{array}\right).
\]
Let us work out a few coactions.
\bea
\Delta^M\left(\trianglebrg\right) & = & \trianglebrg\otimes \triangle+\bubblegbgb\otimes \triangler\nonumber\\
 & &\nonumber\\ & + & \bubblebrbr\otimes \triangleg
+ \bubblegrgr\otimes \triangleb\nonumber\\
 & &\nonumber\\ & + & \tadpolebb \otimes \trianglerg+\tadpolerr\otimes \trianglegb\nonumber\\
  & &\nonumber\\
 & + & \tadpolegg\otimes\trianglebr.\nonumber
\eea
Note that $\Delta^M(M_{i,i})=M_{i,i}\otimes M_{i,i}$ is group-like, so
for example
\[
\Delta^M\left(\triangle\right)=\triangle\otimes\triangle.
\]
Furthermore the tadpole graphs fulfill
\[
\Delta^M\left(\tadpolegg\right)=\tadpolegg\otimes\tadpoleg,
\]
and similar for blue, red.
\end{ex}

This construction is generic and constructs a graphical coaction for any graph with any number of loops and legs. 

For example a $n$-loop vertex graph $G$ in $\Phi^4$ theory has $2n$ edges and all spanning trees $T\in\mathcal{T}(G)$ have $n$ edges.

We get $|\mathcal{T}(G)|\times n!$ Hodge matrices $M(G,T_\mathfrak{o})$,
where  the number of spanning trees
\[
|\mathcal{T}(G)|=\psi(1,\ldots,1),
\]
is given through the first Symanzik polynomial evaluated at unit arguments.

The required graphical coaction $\Delta^G$ then comes as above by summing the individual coactions $\Delta^{M(G,T_\mathfrak{o})}$ for $M(G,T_\mathfrak{o})$ which corresponds to a construction of a matrix $M=M^G$ from all the matrices $M(G,T_{\mathfrak{o}})$.

The situation simplifies if we use kinematical renormalization schems which set tadpole graphs to zero and use (see \cite{MarkoDirk})
\be\label{spacelikeFR}
\Phi_R(G)=\sum_{T\in\mathcal{T}(G)}\vec{\Phi}_R((G,T)).
\ee
Here on the rhs we use Feynman rules $\vec{\Phi}_R((G,T))$ integrating the spacelike parts of loop momenta after the energy components $k_{i,0}$ have been integrated out as residue integrals. These residure integrals generated the sum over spanning trees on the right \cite{MarkoDirk}.

This allows to  erase the leftmost column and uppermost row in the matrices $M^{(i)}$ and  we get six $3\times 3$ matrices $N^{(i)}$ which we can combine to a matrix 
$N$ as follows:
\[
N=\left(
\begin{array}{cc|cc|cc|cc|c}
\left(\tadpoleg+\tadpoler+\tadpoleb\right)\sim 0 & & 0  & & 0  & & 0  & & 0\\
\hline
\underbrace{\bubblegrr +\bubblegrg}_{\bubblegrgr}  & & \bubblegr & & \nnn & & \nnn & & \nnn\\
\hline
\underbrace{\bubblebrb + \bubblebrr}_{\bubblebrbr} & & \nnn & & \bubblebr & & \nnn & & \nnn\\
\hline
\underbrace{\bubblegbb + \bubblegbg}_{\bubblegbgb} & & \nnn & & \nnn & & \bubblegb & & \nnn\\
\hline
\underbrace{\trianglebr +\trianglegb + \trianglerg}_{\trianglebrg} & & \triangleb & & \triangleg & & \triangler & & \triangle\\
\end{array}\right).
\]
Here we could write
\[
\trianglebr +\trianglegb + \trianglerg=\trianglebrg,
\]
as the sum over residues when doing the $dk_{i,0}$ (contour) integrations for any graph $G$ pairs off with the spanning trees of $G$ by Eq.(\ref{spacelikeFR}). 
Here an entry $(G,F)\in H_{GF}$ in the matrix is shorthand for $\Phi_R((G,F))$.

Note that the corresponding coaction $\Delta^N$ is utterly based on Cutkosky graphs:
\beas
\Delta^N\left(\trianglebrg\right) & = & \trianglebrg\otimes \triangle+
\bubblegbgb\otimes \triangler\\
 & &\\
 & + & \bubblebrbr\otimes \triangleg+\bubblegrgr\otimes\triangleb.\\
  & &
\eeas 
Also,
\[
\Delta^N\left(\bubblegbgb\right)=\bubblegbgb\otimes \bubblegb,
\]
and so on. This is particularly useful when using kinematical renormalization schemes where indeed any tadpole vanishes.

There is much more information in our matrices (where we understand that entries are evaluated by renormalized Feynman rules)
\[
N^{(1)}=\left(
\begin{array}{ccccc}
\tadpoleg & & \nnn  & & \nnn\\
\uparrow\pi & & & & \\ 
\bubblegbb & \rightleftharpoons_\mathrm{disp}^\mathrm{Var} & \bubblegb & & \nnn\\
\uparrow\pi & & \uparrow\pi & & \\
\trianglebrb & \rightleftharpoons_\mathrm{disp}^\mathrm{Var} & \triangler & \rightleftharpoons_\mathrm{disp}^\mathrm{Var} & \triangle\\
\end{array}\right).
\]
Some properties:
\begin{itemize}
\item The boundary $d$ of the cubical chain complex \cite{rational} and its action on a graph $G_F\in H_{GF}$ is realized on $\Phi_R(G_F)$ as indicated for $N^{(1)}$ above. 
\[
d=d_0+d_1,\, d\circ d=d_0\circ d_0=d_1\circ d_1=0,
\]
goes to the right: $\mathrm{Var}(\Phi_R(G_F))= \Phi_R(d_0(G_F))$, and up: $\pi \circ \Phi_R(G_F)= \Phi_R(d_1(G_F))$,
corresponding to $m_r$ and $m_c$ in the coaction.
\item Any variation induces a transition in the columns $C_i\to C_{i+1}$ by putting an edge $e$ with
quadric $Q(e)$  on the mass-shell. Therefore Hodge matrices.
\item This determines a point in a fiber determined by the zero locus $Q(e)=0$ 
of the quadric $Q(e)$ assigned to edge $e$. $\pi$ is the corresponding projection 
onto a base space provided by the reduced graph. It also determines  a sequence of iterated integrals associated to the order $\mathfrak{o}$ in either parametric or quadric Feynman rules. The next Sec.(\ref{trianglecomp}) gives an example.
\item Any row $R_{i+1}$ is a fibration over $R_i$ by a one-dimensional fiber.
For example the $z$-integral in Eq.(\ref{Coverbeta}) is an integral over such a one-dimensional fiber.
\item   Boundaries of the dispersion integral are provided by the leading singularities stored in $M_{\mathsf{D}}$.
\end{itemize}
\subsection{The triangle graph}\label{trianglecomp}
Consider the one-loop triangle with vertices $\{A,B,C\}$ and edges 
\[
\{(A,B),(B,C),(C,A)\},
\]
and quadrics (in this example we use both $p,q$ to indicate 4-momenta as we are not invokung parametric variables) :
$$P_{AB}=k_0^2-k_1^2-k_2^2-k_3^2-M_1,$$ 
$$P_{BC}=(k_0+q_0)^2-k_1^2-k_2^2-k_3^2-M_2,$$ 
$$P_{CA}=(k_0-p_0)^2-(k_1)^2-(k_2)^2-(k_3-p_3)^2-M_3.$$
Here, we Lorentz transformed into the rest frame of the external Lorentz 4-vector $q=(q_0,0,0,0)^T$, and oriented the space like part of $p=(p_0,\vec{p})^T$ in the 3-direction: $\vec{p}=(0,0,p_3)^T$.

Using $q_0=\sqrt{q^2}$, $q_0p_0=q_\mu p^\mu\equiv q.p$, $\vec{p}\cdot\vec{p}=\frac{q.p^2-p.pq.q}{q^2}$, we can express everything in covariant form whenever we want to.

We consider first the two quadrics $P_{AB},P_{BC}$ which intersect in $\mathbb{C}^4$.

The real locus we want to integrate is $\mathbb{R}^4$, and we split this as $\mathbb{R}\times\mathbb{R}^3$,
and the latter three dimensional real space we consider in spherical variables as $\mathbb{R}_+\times S^1\times [-1,1]$,
by going to coordinates  $k_1=\sqrt{s}\sin\phi\sin\theta$,$k_2=\sqrt{s}\cos\phi\sin\theta$,
$k_3=\sqrt{s}\cos\theta$, $s=k_1^2+k_2^2+k_3 ^2$, $z=\cos\theta$.

We have 
$$P_{AB}=k_0^2-s-M_1,$$
$$P_{BC}=(k_0+q_0)^2-s-M_2.$$
So we learn say $s=k_0^2-M_1$ from the first
and $$k_0=k_r:=\frac{M_2-M_1-q_0^2}{2q_0}$$ from the second,
so we set
$$s_r
:=\frac{M_2^2+M_1^2+(q_0^2)^2-2(M_1M_2+q_0^2M_1+q_0^2M_2)}{4q_0^2}.$$

The integral over the real locus transforms to 
$$\int_{\mathbb{R}^4}d^4k\to \frac{1}{2}\int_{\mathbb{R}}\int_{\mathbb{R}_+}\sqrt{s} \delta_+(P_{AB})\delta_+(P_{BC})dk_0ds\times \int_0^{2\pi}\int_{-1}^1 d\phi \delta_+(P_{CA})dz.$$
We consider $k_0,s$ to be base space coordinates, while $P_{CA}$ also depends on the fibre coordinate $z=\cos\theta$. Nothing depends on $\phi$ (for the one-loop box it would).

Integrating in the base and integrating also $\phi$ trivially in the fibre  gives
$$\frac{1}{2} \frac{\sqrt{s_r}}{2q_0}2\pi \int_{-1}^1 \delta_+(P_{CA}(s=s_r,k_0=k_r))dz.$$

For $P_{CA}$ we have 
\be\label{alphabeta}
P_{CA}=(k_r-p_0)^2-s_r -\vec{p}\cdot\vec{p}-2|\vec{p}|\sqrt{s_r}z-M_3=:\alpha+\beta z.
\ee
Integrating the fibre gives a very simple expression (the Jacobian of the $\delta$-function is $1/(2\sqrt{s_r}|\vec{p}|)$, and we are left with
the Omn\`es factor\footnote{For any 4-vector $r$ we have $r^2=r_0^2-\vec{r}\cdot\vec{r}$. Let $q$ be a time-like 4-vector, $p$ an arbitrary 4-vector.
Then, $(q\cdot p^2-q^2 p^2)/q^2=\lambda(q^2,p^2,(q+p)^2)/4q^2$
and in the rest frame of $q$, $(q\cdot p^2-q^2 p^2)/q^2=\vec{p}\cdot \vec{p}$ where $\lambda(a,b,c)=a^2+b^2+c^2-2(ab+bc+ca)$, as always.}
\be\label{Omnes}
\frac{\pi}{4|\vec{p}|q_0}=\frac{\pi}{2\sqrt{\lambda(q^2,p^2,(q+p)^2)}}
=\Phi_R\left(\triangle\right).
\ee

This contributes as long as the fibre variable 
\be\label{fiberz}
z=\frac{(k_r-p_0)^2-s_r -\vec{p}\cdot\vec{p}-M_3}{2|\vec{p}|\sqrt{s_r}}
\ee
 lies in the range $(-1,1)$.
This is just the condition that the three quadrics intersect.

An anomalous threshold below the normal theshold appears when $(m_1-m_2)^2<q^2<(m_1+m_2)^2$.

On the other hand, when we leave the propagator $P_{CA}$ uncut,
we have the integral
$$\frac{1}{2} \frac{\sqrt{s_r}}{2q_0}2\pi \int_{-1}^1 \frac{1}{P_{CA}}_{(s=s_r,k_0=k_r)}dz.$$
This delivers a result as foreseen by $S$-Matrix theory \cite{Polkinghorne,ELOP}.

The two $\delta_+$-functions constrain the $k_0$- and $t$-variables, so that the remaining integrals are over the compact domain $S^2$. Here the fiber is provided by the one-dimensional $z$-integral and the compactum $C_{G/E_F}$ is the two-dimensional $S^2$ while for $C_G$ it is the one-dimensional $S^1$.  

As the integrand does not depend on $\phi$, this gives a result of the form
\bea\label{Coverbeta} 
\Phi_R\left(\triangler\right) & = & 2\pi C \underbrace{\int_{-1}^1 \frac{1}{\alpha+\beta z}dz}_{=:J_{CA}}=2\pi \frac{C}{\beta}\ln\frac{\alpha+\beta}{\alpha-\beta}\\ & = &\frac{1}{2}\underbrace{\mathrm{Var}(\Phi_R(b_2))}_{\Phi_R\left(\bubblegbgb\right)}\times J_{CA},\nonumber
\eea
where $C=\sqrt{s_r}/2q_0$ is intimitaly related to $\mathrm{Var}(\Phi_R(b_2))$
for $b_2$ the reduced triangle graph (the bubble), and the factor $1/2$ here is $
\mathrm{Vol}(S^1)/\mathrm{Vol}(S^2)$.

Here, $\alpha$ and $\beta$ are given through (see Eq.(\ref{alphabeta}))
$l_1\equiv \vec{p}^2=\lambda(q^2,p^2,(p+q)^2)/4q^2$ and $l_2:=s_r=\lambda(q^2,M_1,M_2)/4q^2$ as
\[
\alpha=(k_r-p_0)^2-l_2 -l_1-M_3,\,\beta=2\sqrt{l_1l_2}.
\]
Note that
\[
\frac{C}{\beta}=\frac{1}{\sqrt{\lambda(q^2,p^2,(q+p)^2)}}=\frac{1}{2q_0|\vec{p}|},
\]
in Eq.(\ref{Coverbeta}) is proportional to the Omn\`es factor Eq.(\ref{Omnes}).

In summary, there is a Landau singularity in the reduced graph in which we shrink $P_{CA}$. It is  located at 
\[
q_0^2=s_{normal}=(\sqrt{M_1}+\sqrt{M_2})^2=s_{\left(\bubblegb\right)}.
\]
It corresponds to the threshold divisor defined by the intersection\\ $(P_{AB}=0)\cap (P_{BC}=0)$ at the point
\[
\left(q^2=(\sqrt{M_1}+\sqrt{M_2})^2, \underbrace{\frac{A_1}{A_2}=\frac{\sqrt{M_2}}{\sqrt{M_1}}}_{{\text{barycenter}}\left(\bubblegb\right)}\right)
\]

 This is not a Landau singularity
when we unshrink $P_{CA}$ though. A (leading) Landau singularity appears
in the triangle when we also intersect the previous divisor with the locus $(P_{CA}=0)$.

It has a location which can be computed from the parametric approach.   
One finds
\bea\label{sanom} 
q_0^2 & = & s_{anom}  =  (\sqrt{M_1}+\sqrt{M_2})^2+\nonumber\\
& & +\frac{4M_3(\sqrt{\lambda_2}\sqrt{M_1}-\sqrt{\lambda_1}\sqrt{M_2})^2-\left(\sqrt{\lambda_1}(p^2-M_2-M_3)\right.}{4M_3\sqrt{\lambda_1}\sqrt{\lambda_2}}\nonumber\\
 & & +\frac{\left.\sqrt{\lambda_2}((p+q)^2-M_1-M_3)\right)^2}{4M_3\sqrt{\lambda_1}\sqrt{\lambda_2}}\nonumber\\
 & = & s_{\left(\triangle\right)},\nonumber  
\eea
with $\lambda_1=\lambda(p^2,M_2,M_3)$ and $\lambda_2=\lambda((p+q)^2,M_1,M_3)$. 

Eq.(\ref{Coverbeta}) above is the promised result: the leading singularity of the reduced graph $t/P_{CA}$ and the non-leading singularity of $t$ have the same location and both involve $\mathrm{Var}(\Phi_R(b_2))$ and the non-leading singularity of $t$ factorizes into the (fibre) amplitude $J_{CA}\times \mathrm{Var}(\Phi_R(b_2))$. 

This gives rise to a cycle which is a generator in the cohomology
of the cubical chain complex as Fig.(\ref{cycle}) demonstrates \cite{MarkoDirk}.
\begin{figure}[H]
\includegraphics[width=12cm]{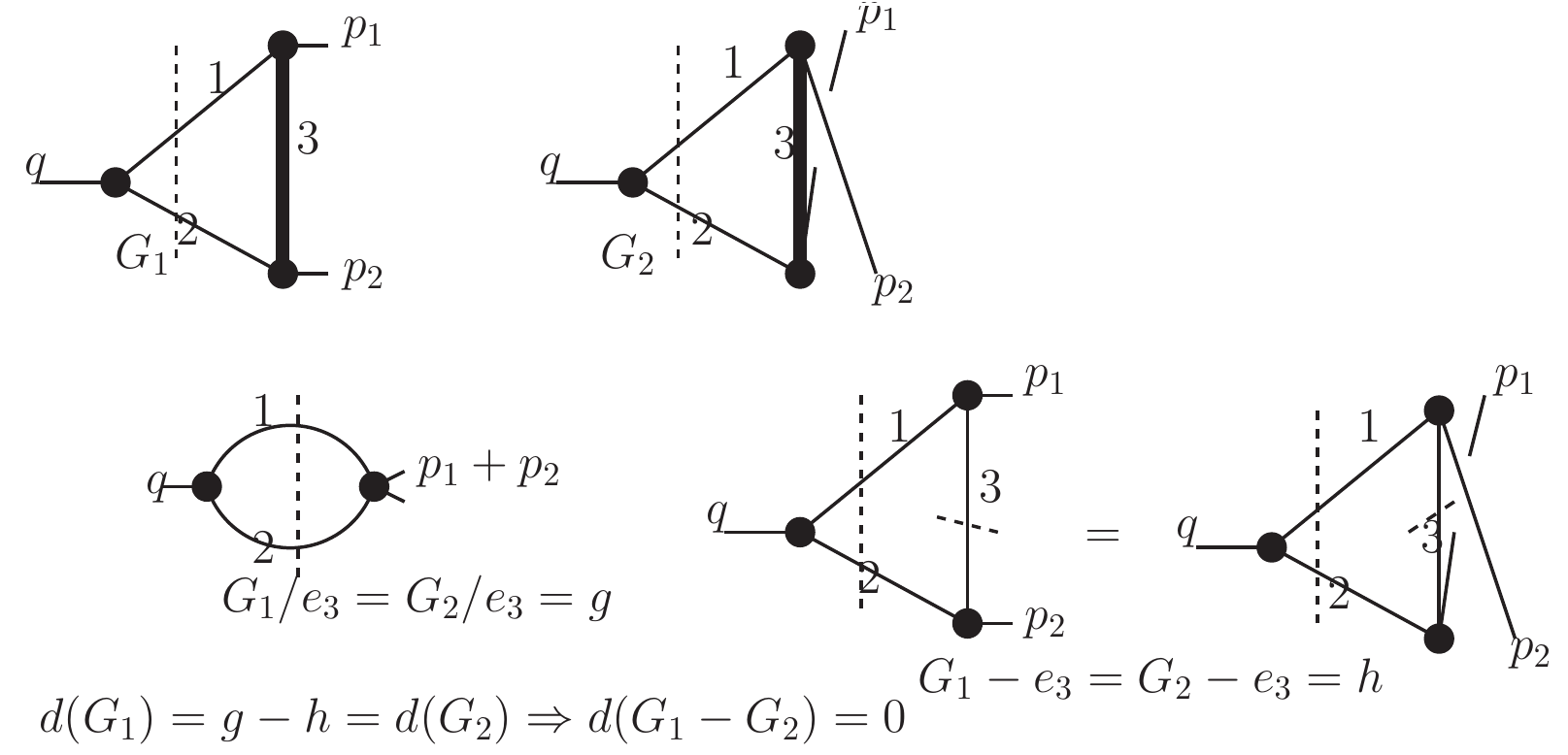}.
\caption{The two Cutkosky  triangle graphs $G_1,G_2$  are distinguished by a permutation of external edges $p_1,p_2$. Edges $e_1,e_2$ are on-shell, $e_3$ is off-shell and hence in the forest. Shrinking or removing it delivers in both cases the same reduced ($g$) or leading ($h$) graph. As a result we get a cycle $d(G_1-G_2)=0$. Obviously there is no $X$ such that $dX=G_1-G_2$.}
\label{cycle}
\end{figure}
As for dispersion, we get a result effectively mapping $C_3\to C_2\to C_1$:
\beas
\Phi_R\left(\trianglebrg\right) & = & \int_{s_{\left(\triangle\right)}}^{s_{\left(\bubblegb\right)}}\frac{\Phi_R\left(\triangle\right)}{s-x}dx\\
& + & \int_{s_{\left(\bubblegb\right)}}^{\infty}\frac{\Phi_R\left(\triangler\right)}{s-x}dx.
\eeas
The situation is very similar for the Dunce's cap graph $dc$.
Again we have spanning trees of length two and monodromy generated from partitioning its three vertices in all possible ways by cuts.  

Look first at a single term for a chosen ordered spanning tree:
\begin{figure}[H]
\includegraphics[height=40mm]{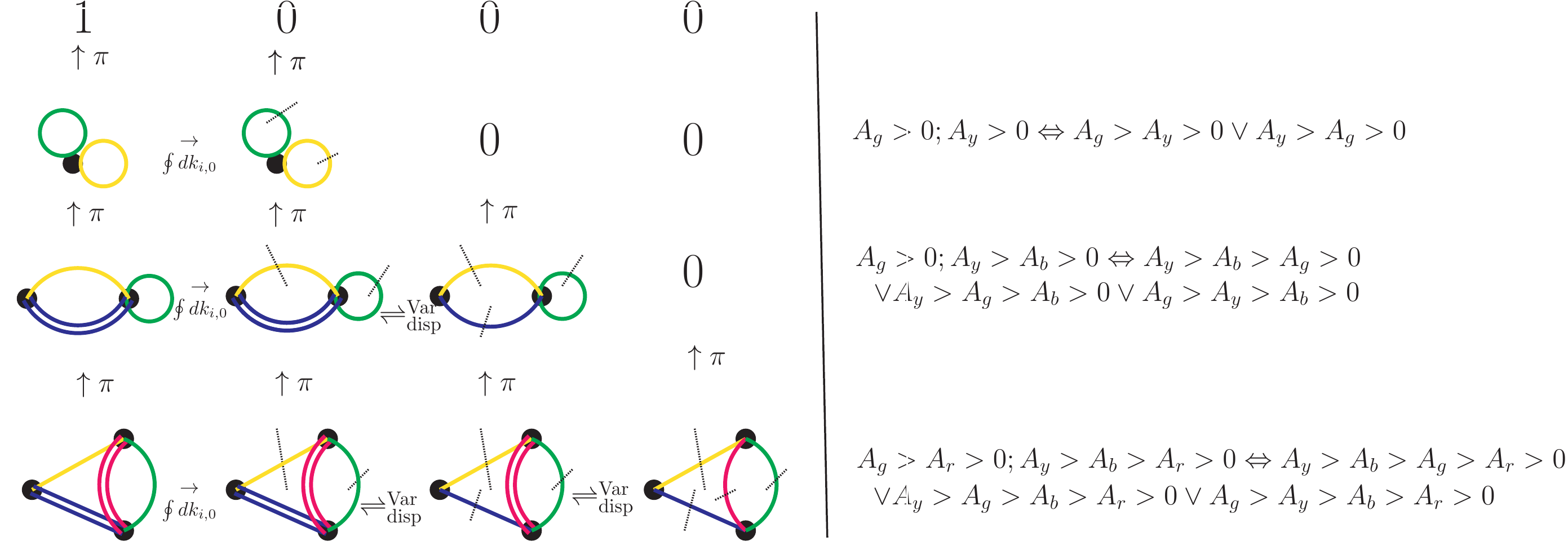}
\caption{Dispersion in the Dunce's cap. Here the order is blue before red, so red shrinks first and blue is cut first.
Note that due to the presence of more than one loop, choosing a spanning tree (blue, red: the thick double edges) and an order does not single out a single sector as it would in the one-loop case. Here we get three sectors. See the right column.
Summing over trees and orders correctly delivers all $24$ sectors from the ten ordered spanning trees partitioning them as $24=3+2+3+2+3+2+3+2+2+2$ as we see below.}
\label{dccoexDisp}
\end{figure}
\subsection{Summing up}
We use Eq.(\ref{spacelikeFR}) where $\vec{\Phi}_R$ has integrated out all energy integrals $\oint \prod_{i=1}^{|G|} dk_{i,0}$
by contour integrations closing in the upper halfplane.

This leads to a graphical coaction:
\begin{thm}
\[
\Delta^G=\sum_{(T,\mathfrak{o})\sim G}\Delta^{G_{T_\mathfrak{o}}}(g),
\]
defines a graphical coaction for all $g\in \mathsf{Gal}(M^G)$. For kinematical renormalization schemes 
it can be written as a coaction on Cutkosky graphs.
\end{thm}
\begin{cor}
Assume the number of spanning trees equals the number of edges of a graph,  $|\mathcal{T}(G)|=e_G$ which is true for one-loop graphs and their duals, banana graphs. We call them simple graphs (in blunt ignorance of the analytical complexity of an $n$-edge banana graph, $n\geq 3$).
Then
\[
\sum_{(T,\mathfrak{o})\sim G}\Delta^{G_{T_\mathfrak{o}}}(G)=\Delta_{\mathsf{Inc}}(G),
\]
where $\Delta_{\mathsf{Inc}}$ is the incidence Hopf algebra and coaction used by Britto et.\ al.\ \cite{Brittoetal}.
\end{cor}
Here,\[
\Delta_{\mathsf{Inc}}(G)=\Delta_{\mathsf{Inc}}(G,\emptyset)=
\sum_{\emptyset\leq X\leq E_G\atop X\not= \emptyset}(G_X,\emptyset)\otimes (G,X),
\]
in their notation. $G_X$ has all edges $e$ contracted, $e\not\in X$.
In fact one-loop graphs evaluate to dilogs \cite{BlKr} and hence provide the first examples to pull back the coaction from such functions to graphs.

There are non simple graphs in dedicated kinematics (massless internal edges, lightlike external momenta) where $\Delta_{\mathsf{Inc}}$ agrees
with $\Delta^G$ as well, but not in a generic situation:
\begin{cor}
The first non simple graph is the Dunce's cap graph $dc$ with two loops and three vertices.
It has four edges and five spanning trees.
$\Delta_{\mathsf{Inc}}(dc)\not=\Delta^{dc}(dc)$.
\end{cor}
Similar for all other non simple graphs in generic kinematics.
\subsection{The Dunce's cap}
\begin{ex}
Let us work out the Dunce's cap. We start with Fig.(\ref{TetraDunce}).
\begin{figure}[H]
\includegraphics[width=12cm]{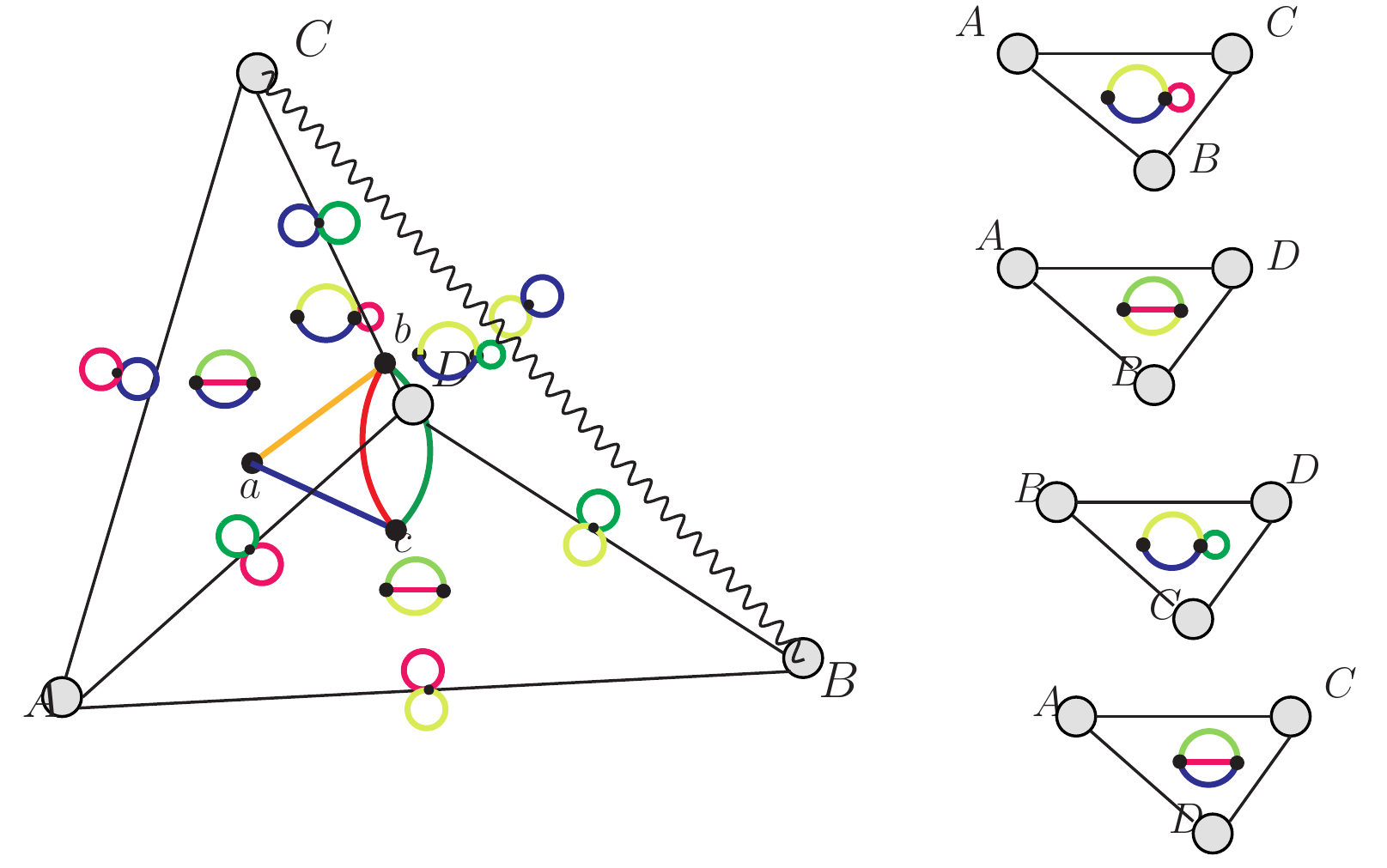}
\caption{The Dunce's cap $dc$ and its cell $C(dc)$, a tetrahedron. We also indicate the four triangular cells which are its co-dimension one hypersurfaces. It is a graph on four edges, its cell in OS is thus the three-dimensional tetrahedron $C(dc)$. Its spine gives rise to five two-dimensional cubes $Q_i$ which can not provide a triangulation of $C(dc)$. Instead $C(dc)$ gives rise to a fibration of the cubes 
$Q_i$. The spine is a union of ten paths. Six of them give rise to two sectors, and four of them to three sectors, adding up to the $24$ sectors in $C(dc)$. Renormalization makes the extra sector in the latter four paths well-defined.}
\label{TetraDunce}
\end{figure}
There are $6=(\genfrac{}{}{0pt}{}{4}{2})$ choices 
for two out of four edges. One of these does not form a possible basis for two loops in the graph, the other five choices determine the five spanning trees of the graph
as in Fig.(\ref{duncetrees}). Correspondingly the co-dimension two edge $BC$ is not part of the cell of the Dunce's cap, nor are the four corners.
\begin{figure}[H]
\includegraphics[width=10cm]{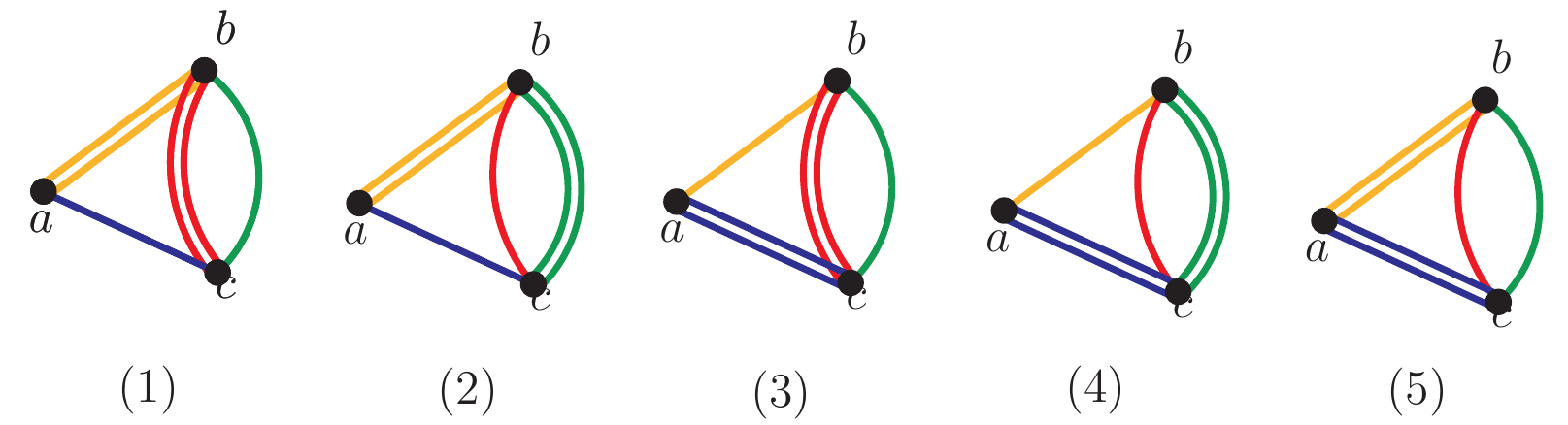}
\caption{The five spanning trees of the Dunce's cap $dc$. They give rise to five cubes $Q_i$ and ten matrices $M(dc,T_{\mathfrak{o}})$ . Spanning trees are on two edges so we get two possible orders and hence ten matrices. Integrating the energy variables indeed generates residues $\sum_T \vec{\Phi}_R(G_T)$ for those trees.}
\label{duncetrees}
\end{figure}
We give one cube as an example in Fig.(\ref{dccube}).
\begin{figure}[H]
\includegraphics[width=10cm]{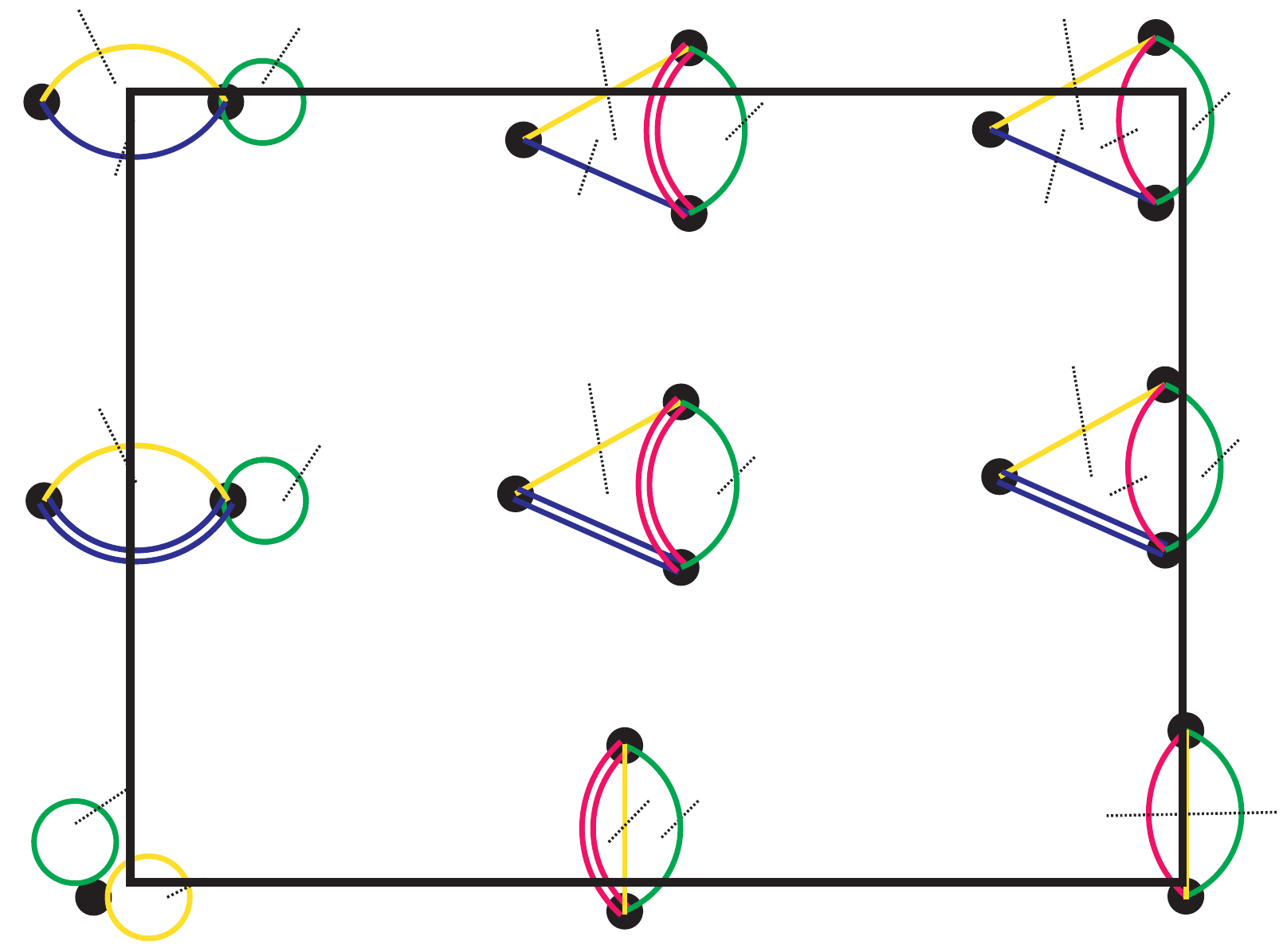}
\caption{A cube $Q$ for the graph $dc$. It gives rise to two matrices $M^{(j)}$.
Note that it contains the six entries of $M^{(3)}$. Note that in all nine entries of the cube graphs are evaluated at $A_gm_g=A_ym_y$ and the cube  describes a codimension 1 surface of $\mathbb{P}^3=\mathbb{P}_{dc}$. The one-dimensional fibre which has the cube as base is given by the variable $A_y/A_g$. }
\label{dccube}
\end{figure}

\begin{figure}[H]
\includegraphics[width=2cm]{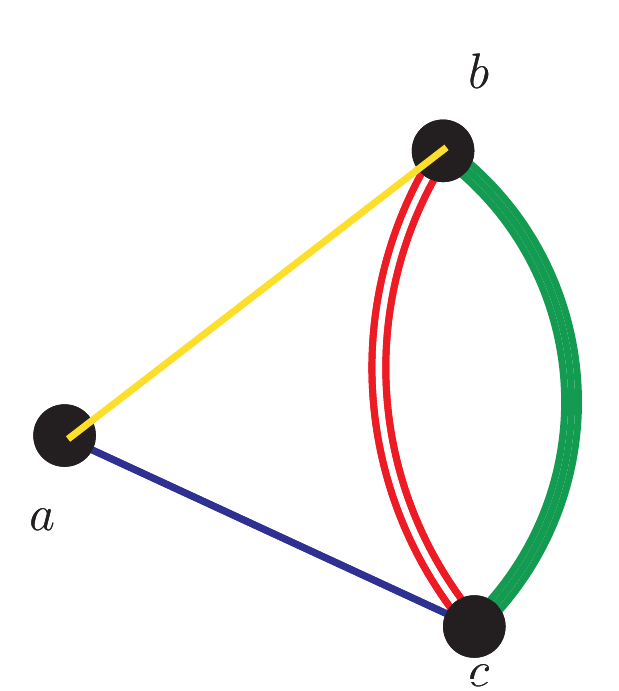}
\caption{This is illegal. The green and red edge do not form a spanning tree. Correspondingly there is neither matrix nor residue assigned to this configuration and hence for the nonsimple Dunce's cap graph the coaction $\Delta_{\mathsf{Inc}}$  of \cite{Brittoetal} (which includes this graph) deviates from the structure of a cubical complex.}
\label{duncetreeillegal}
\end{figure}
\end{ex}
For example for the spanning tree $T_3$ with order blue before red so that we shrink red first we find the matrix $M^{(3)}$ given in Fig.(\ref{dccoex}).
\begin{figure}[H]
\includegraphics[width=11cm]{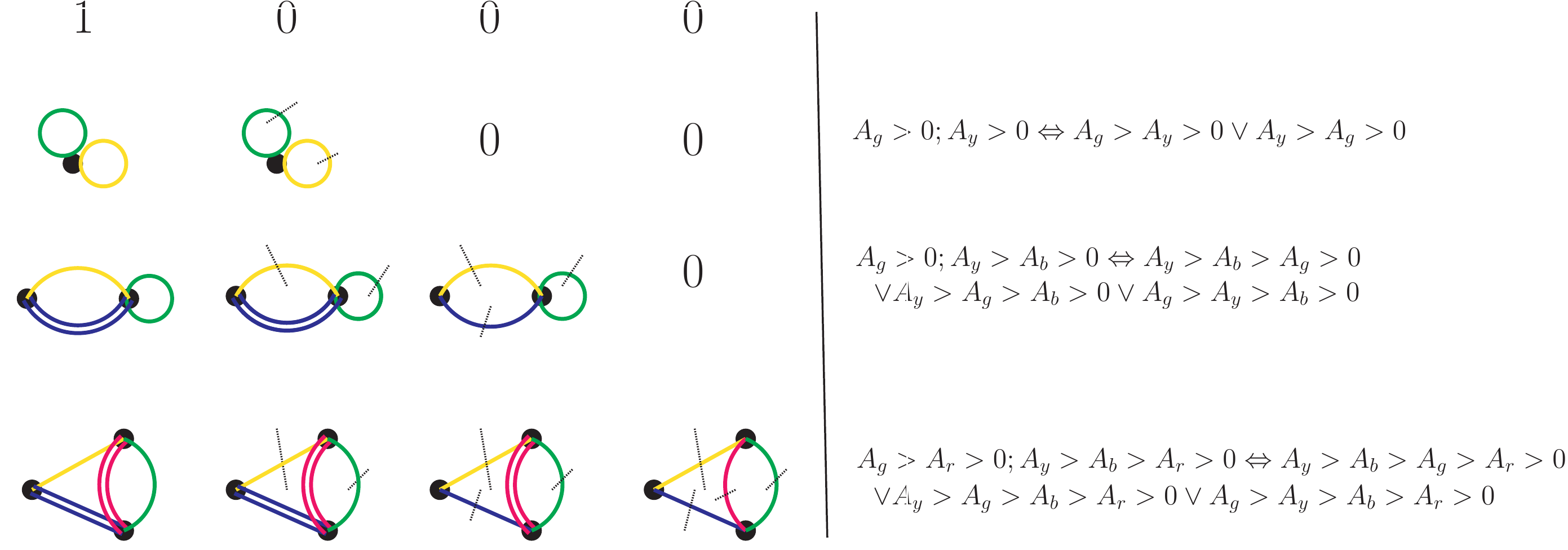}
\caption{The matrix $M^{(3)}$ which we had before. We have obviously four such matrices giving three sectors each.}
\label{dccoex}
\end{figure}
Applying $\Delta^M$ is in Fig.(\ref{dccoact}).
\begin{figure}[H]
\includegraphics[width=12cm]{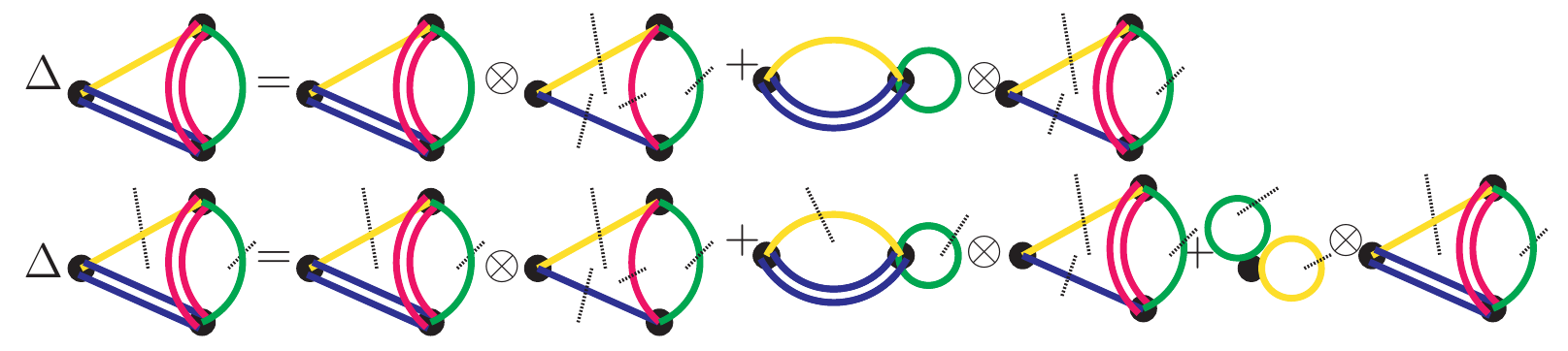}
\caption{In the upper row $\Delta^M$ acts as a coaction, in the lower as a coproduct.}
\label{dccoact}
\end{figure}
If we change the order to red before blue we get a different matrix:
\begin{figure}[H]
\includegraphics[width=11cm]{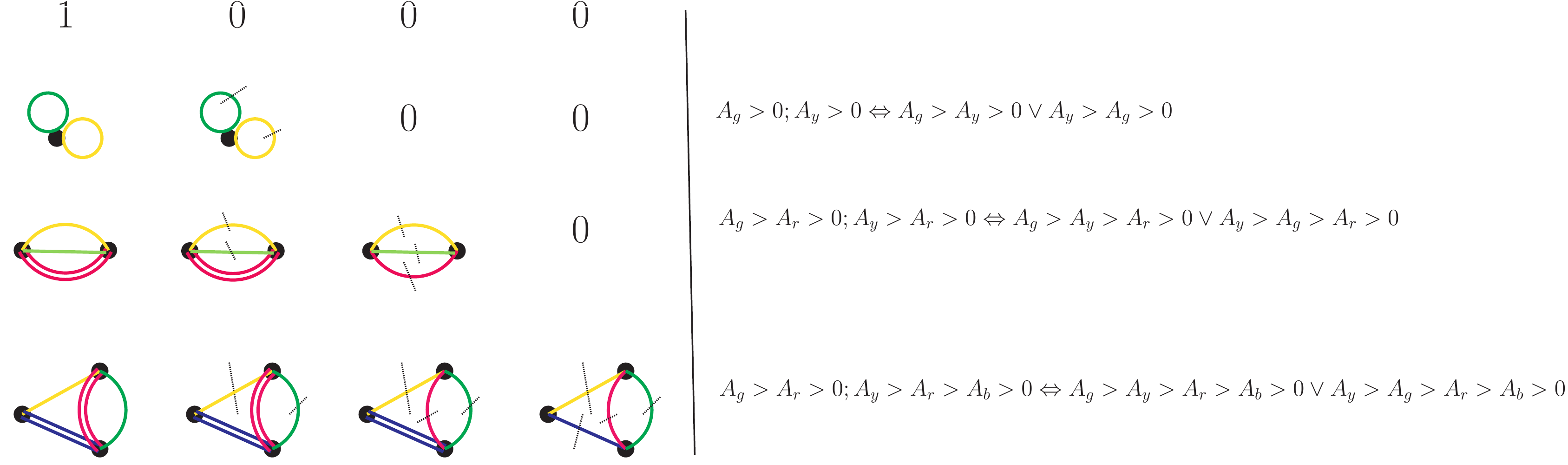}
\caption{The matrix $M^{(2)}$. We also give the sectors to which its entries contribute of which there are two and we have four such matrices.}
\label{dccoextwo}
\end{figure}
Finally the case of a spanning tree on the yellow and blue edge, with order yellow before blue:
\begin{figure}[H]
\includegraphics[width=11cm]{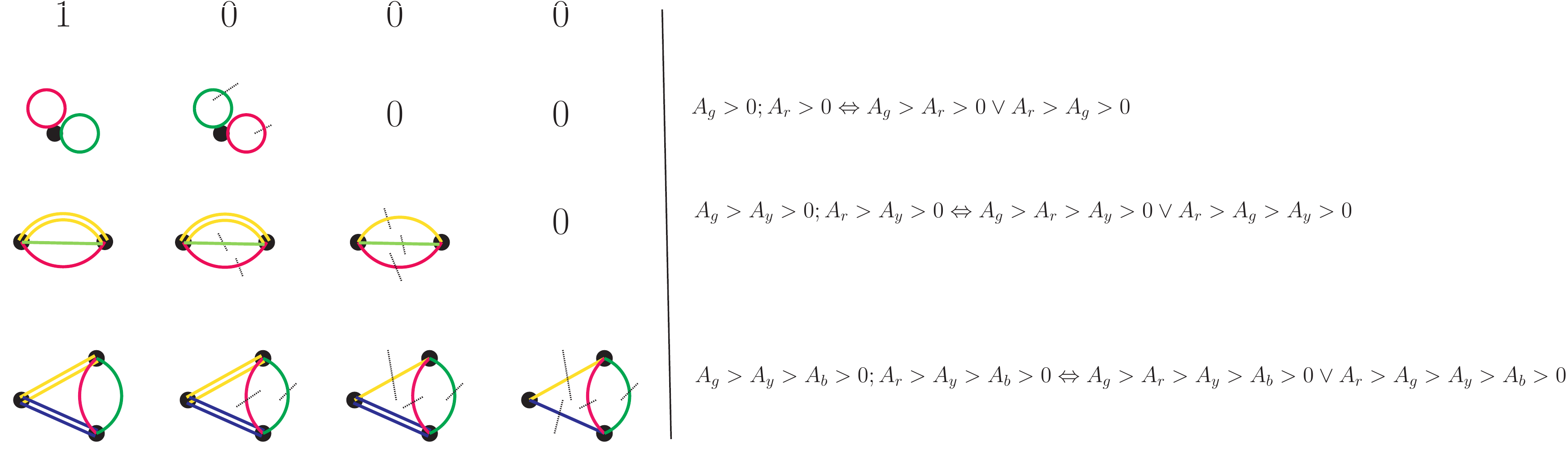}
\caption{The matrix $M^{(9)}$. We also give the sectors to which its entries contribute of which there are two and we have two such matrices from the two possible orders.}
\label{dccoexonine}
\end{figure}
Next we can get rid of dangling tadpole graphs using for example $M^{(3)}_{\mathsf{D}}$ in $M^{(3)}$ using the matrix of Fig.(\ref{dccoexD}):
\begin{figure}[H]
\includegraphics[width=4cm]{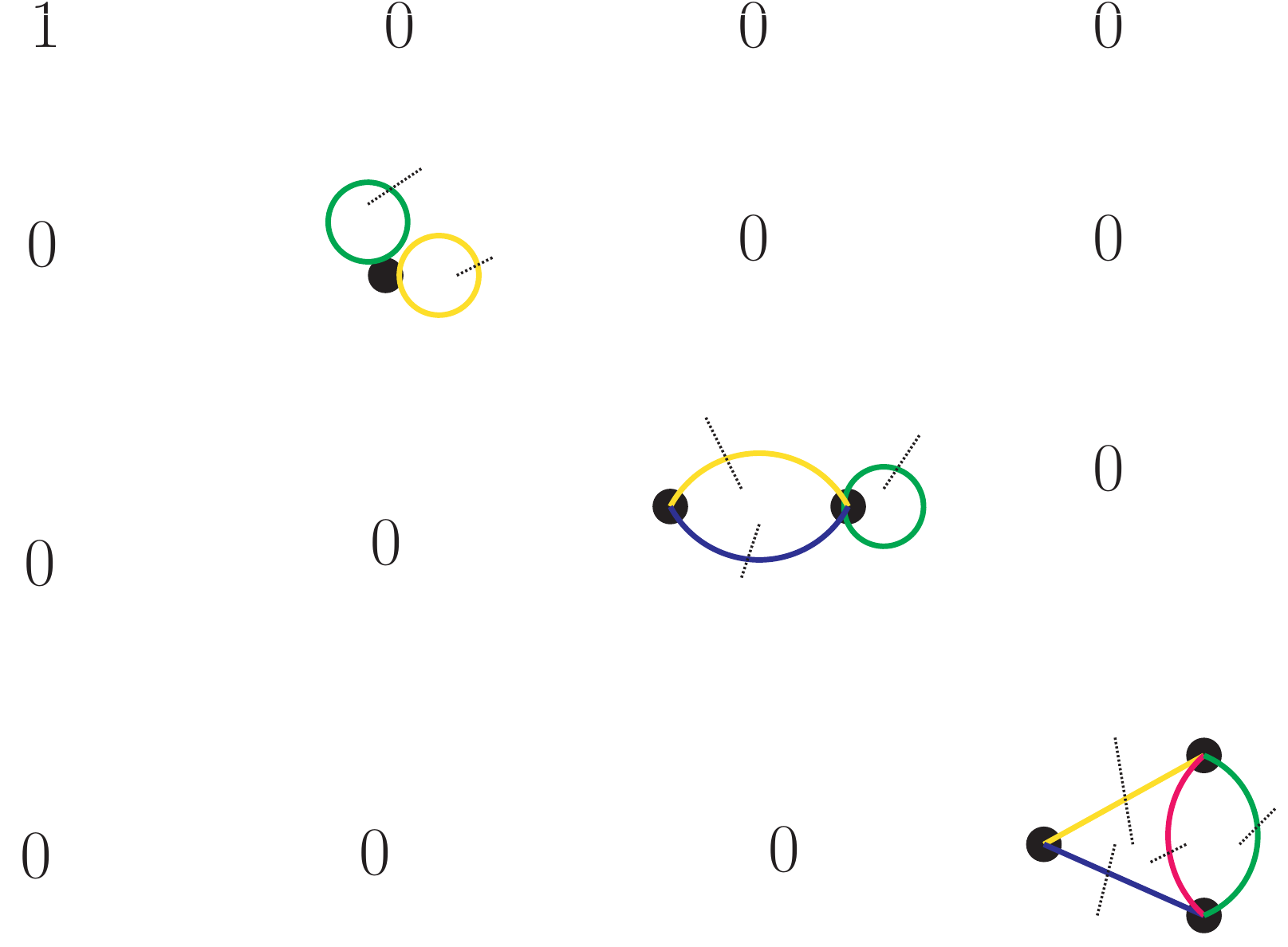}
\caption{The matrix $M^{(3)}_{\mathsf{D}}$. Multiplying from the left with its inverse unifies the diagonal and eliminates all tadpoles due to Eq.(\ref{tadpolevanish}).}
\label{dccoexD}
\end{figure}
and also use the matrix Fig.(\ref{dccoexDp}):
\begin{figure}[H]
\includegraphics[width=4cm]{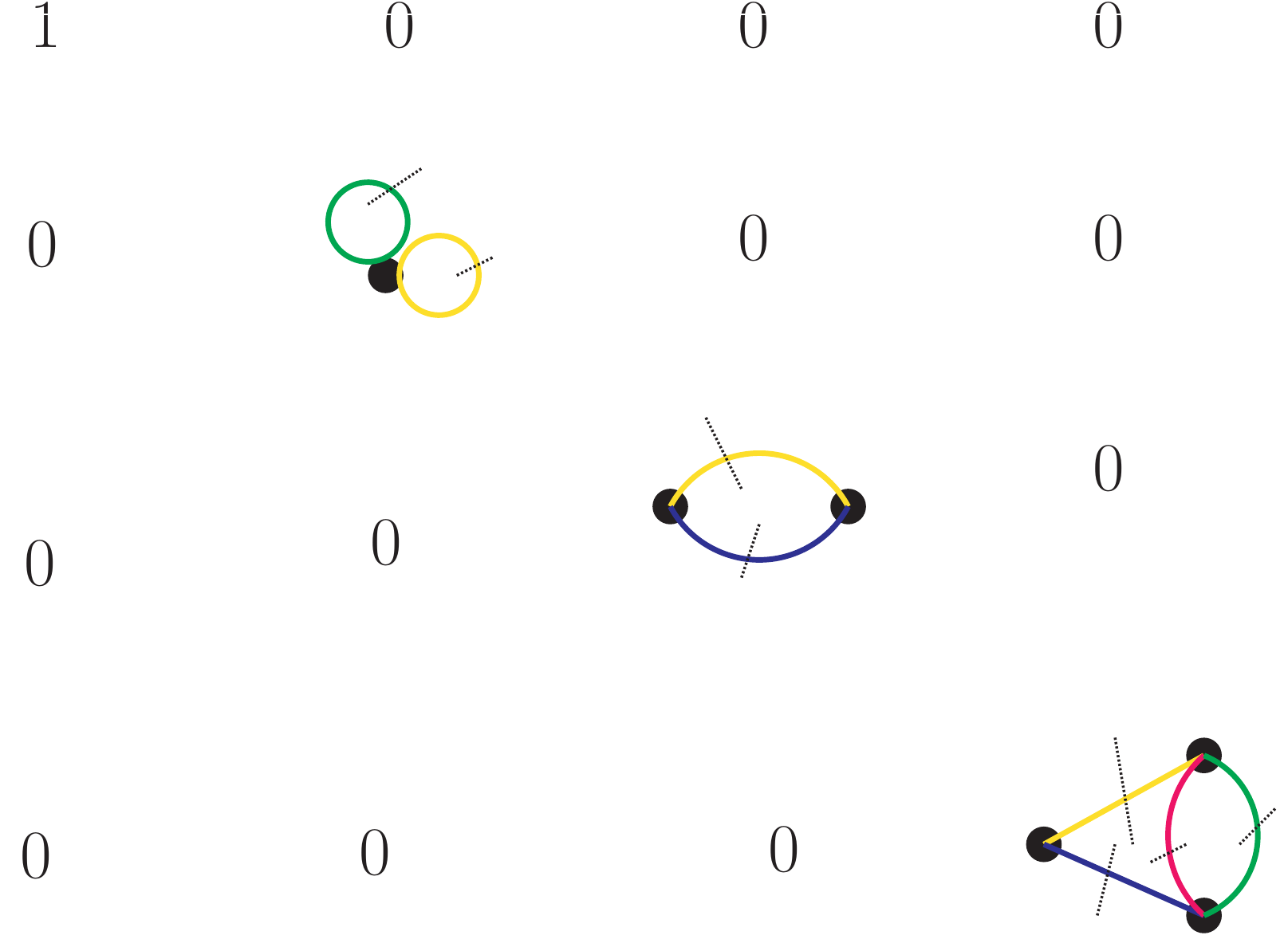}
\caption{The matrix $M^{(3)}_{\mathsf{D}^\prime}$. Multiplying it from the right reinserts all diagonal entries apart from tadpoles.}
\label{dccoexDp}
\end{figure}
We construct $\tilde{M}^{(3)}=M^{(3)}_{\mathsf{D}^\prime}\times
(M^{(3)}_{\mathsf{D}})^{-1}\times M^{(3)}$.
We now sum over orders and spanning trees for all $\tilde{M}^{(i)}$, and use hence kinematic renormalization schemes for which we have
\be\label{tadpolevanish}
\Phi_R\left(\tadpolegg\right)=\Phi_R\left(\tadpoleg\right)=0.
\ee
This then allows to eliminate the leftmost column and topmost row from the coaction matrices and allows to sum over spanning trees so that we can formulate the coaction on Cutkosky graphs.

Again we find a matrix $N=N(dc)$  which defines a coaction which only involves Cutkosky graphs as in Fig.(\ref{dccoN}):
\begin{figure}[H]
\includegraphics[width=8cm]{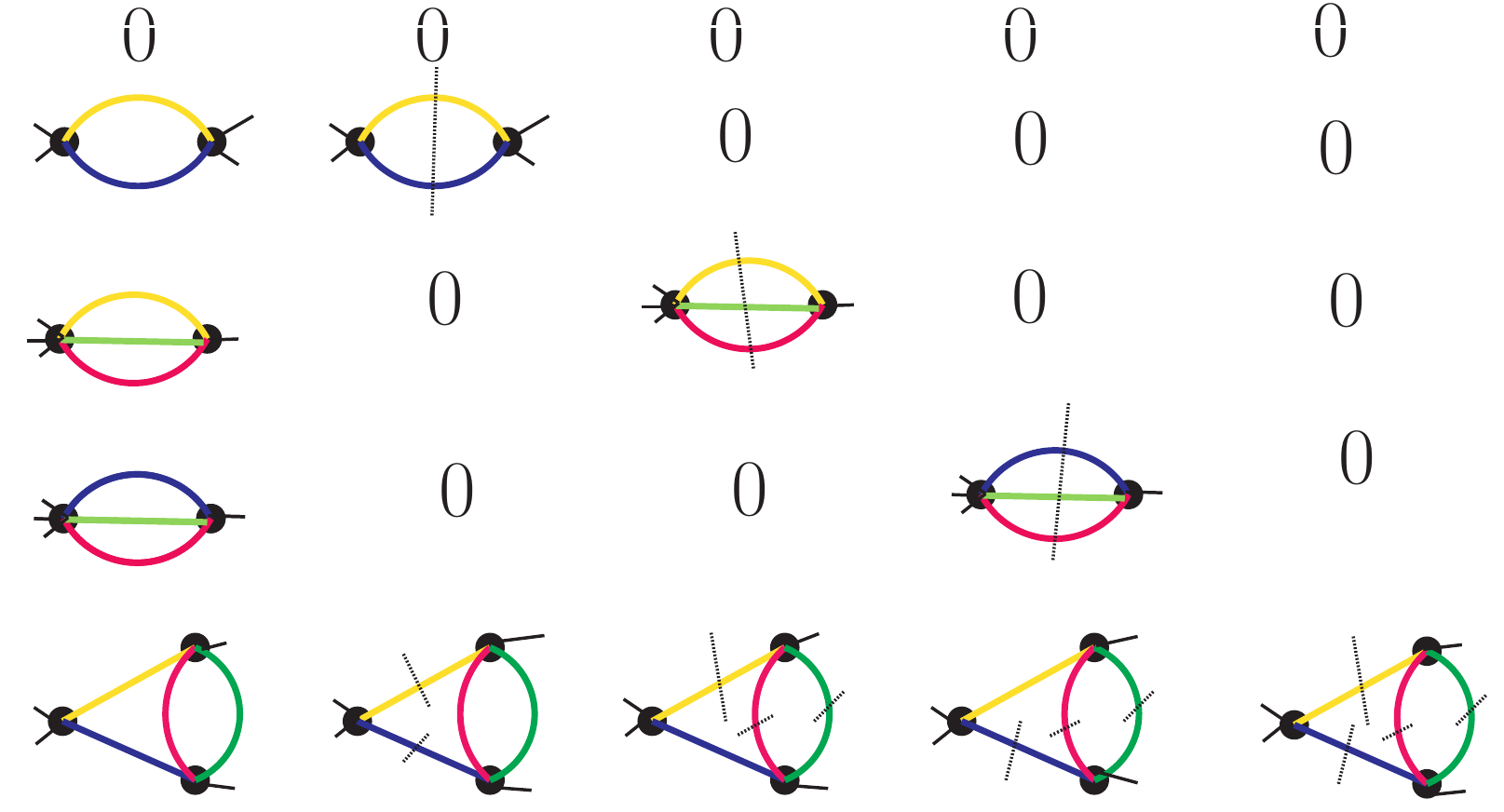}
\caption{The matrix $N(dc)$. The second entry in the lowest row is a shorthand given in Fig.(\ref{shorthand}).}
\label{dccoN}
\end{figure}
\begin{figure}[H]
\includegraphics[width=8cm]{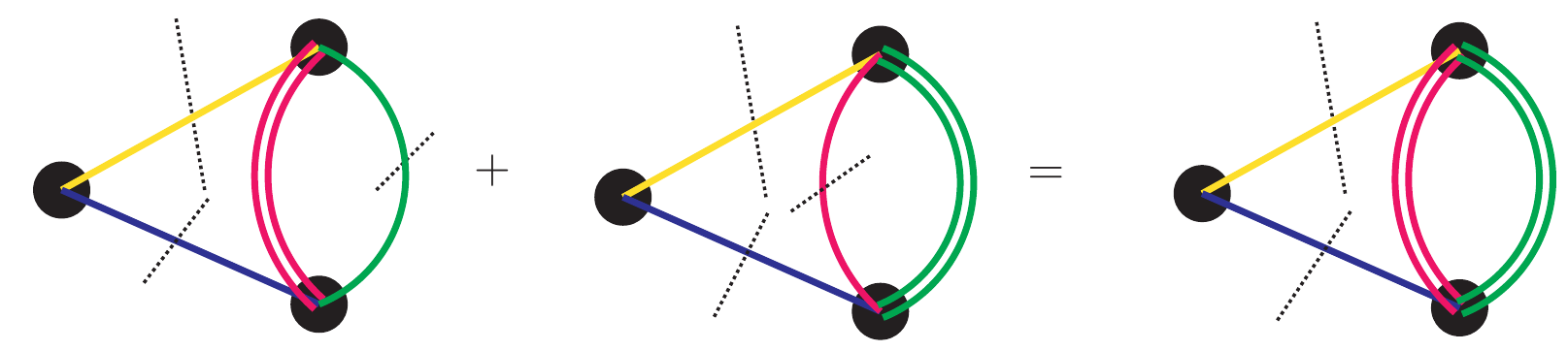}
\caption{Integrating the subloop sums two residues by putting either the red or green edge on shell. We can combine this into one entry in the matrix $N(dc)$ thanks to the fact that tadpoles vanish.}
\label{shorthand}
\end{figure}

\begin{rem} Deformed coactions. Pulling back the known coaction
of (elliptic) polylogs to a graphical coaction Britto et.al.\ find the need to deform their coaction in a systematic way using the parity of the number of edges. We can incorporate this in $\Delta^G$ in a similar fashion but attempt at an approach using the $\mathbb{Z}_2$ grading of graph homology in future work.
\end{rem}
\begin{rem} First entry condition. Steinman relations.
Note that to any entry $M_{j,2}$ belongs a 2-partition
$V_G=V_G^{(1)}\amalg V_G^{(2)}$. This defines a variable $s=(\sum_{v\in V_G^{(1)}}q(v))^2$. The Matrix $M$ then describes the monodromy 
of functions $\Phi_R(M_{j,1})$ in the leftmost column  through the entries
$M_{j,2}$ in the next column when varying this variable $s$.
$M_{j,2}$ are by construction the first entries which have a non-trivial cut
each originating from a distinct non-overlapping  sector.

One interpretation of the Steinmann relation is that two different 2-partitions which define two different variables $s,t$ indeed  do not interfere.
The monodromy in a chosen variable $s$ is solely determined by subdividing the associated 2-partition further.  
\end{rem} 
\section{Conclusions}
\begin{itemize}
\item The cell decomposition of OS together with the orresponding spine
provide a cubical complex for Feynman graphs organized by spanning trees.
\item Boundaries correspond to either reduced or Cutkosky-cut graphs.
\item Each cube has an accompanying simplex decomposition giving Hodge matrices according to a chosen order of edges in a spanning tree.
\item Each Hodge matrix defines a coaction.
\item Summing over trees and orders defines a coproduct and graphical coaction $\Delta^G$  for any Feynman graph $G$.
\item Only in simple cases it agrees with $\Delta_{\mathsf{Inc}}$.
For generic kinematics $\Delta_{\mathsf{Inc}}$ is maximally wrong.
\item The use of dimensional regularization is neither necessary nor sufficient to find a valid graphical coaction.
\item Task: interprete $\Delta^G$ in terms of Brown's approach \cite{BrownI,BrownII}  in particular on the possibly not so mysterious rhs (the 'de Rham side').
\item Question: What is Brown's  small graphs principle making out of the simplifications in kinematic renormalization?
\item This so far is a story on principal sheets and variations in the real domain. For a complete understanding in algebraic geometry one must make room for complex variations of masses and kinematics. Need to take into account finer structure of OS. Whilst here we worked with the spine of OS, 
one needs to consider markings and bordification of OS itself.
\end{itemize}

\begin{acknowledgement}
It is a pleasure to thank Marko Berghoff, Spencer Bloch, Karen Vogtmann and Karen Yeats for helpful advice. I thank Spencer in particular for numerous discussions on the 
matrices investigated here.  I am grateful to Johannes Bl\"umlein and all the organizers
for their efforts.
\end{acknowledgement}
\section*{Appendix 1: The cubical chain complex}\label{ccc}
\addcontentsline{toc}{section}{Appendix 1}
We assume the reader is familiar with the notion of a graph and of spanning trees and forests. See \cite{MarkoDirk} where these notions are reviewed. We follow the notation there. In particular $|G|$ is the number of independent cycles of $G$, $e_G$ the number of internal edges and $v_G$ the number of vertices of $G$. For a pair of a graph and a spanning forest we write 
$(G,F)$ or $G_F$. If a spanning forest has $k$ connected components we call it a $k$-forest. A spanning tree $T$ is a 1-forest. Its number of edges hence $e_T$. $\mathcal{T}(G)$ is the set of spanning trees of $G$.

Pairs $(G,F)$ are elements of a Hopf algebra $H_{GF}$ based on the core Hopf algebra $H_{core}$ of bridgefree graphs \cite{MarkoDirk}. 

As an algebra $H_{GF}$  is the free commutative $\mathbf{Q}$-algebra generated by such pairs. Product is disjoint union and the empty graph and empty tree provide
the unit. 

A $k$-cube is a $k$-dimensional cube $[0,1]^k\subsetneq \mathbb{R}^k$.

Consider $G_T$. We define a cube complex for $e_T$-cubes $\mathsf{Cub}_G^T$ 
assigned to $G$. There are $e_T!$ orderings $\mathfrak{o}=\mathfrak{o}(T)$
which we can assign to the internal  edges of $T$.

We define a boundary for any elements $G_F\equiv (G,F)$ of $H_{GF}$. For this consider such an ordering 
\[
\mathfrak{o}:E_F\to
[1,\ldots,e_F]
\]
of the $e_F$ edges of $F$.
There might be other labels assigned to the edges of $G$
and we assume that removing an edge or shrinking an edge will
not alter the labels of the remaining edges. In fact the whole Hopf algebra structure of $H_{core}$ and $H_{GF}$
is preserved for arbitrarily labeled graphs \cite{Turaev}.  

The (cubical) boundary map $d$ is defined by $d:=d_0+d_1$
where
\begin{equation}\label{eq:cubediff}
d_0(G_F^{\mathfrak{o}(F)}):= \sum_{j=1}^{e_F} (-1)^j(G_{F \setdiff e_j}^{\mathfrak{o}(F\setdiff e_j)}), \quad d_1 (G_F^{\mathfrak{o}(F)}) := \sum_{j=1}^{e_F} ({G/e_j}^{\mathfrak{o}(F/e_j)}_{F/e_j}).
\end{equation} 
We understand that all edges $e_k,k\gneq j$ on the right are
relabeled by $e_k\to e_{k-1}$ which defines the corresponding $\mathfrak{o}(T/e_j)$ or $\mathfrak{o}(T \setdiff e_j)$.
Similar if $T$ is replaced by $F$.

From \cite{rational} we know that $d$ is a boundary:
\begin{thm}\cite{rational}
\[
d\circ d=0,\,d_0\circ d_0=0,\, d_1\circ d_1=0.
\]
\end{thm}

Starting from $G_T$ for any chosen $T\in\mathcal{T}(G)$
each chosen order $\mathfrak{o}$ defines one of $e_T!$ simplices of a $e_T$-cube $\mathsf{Cub}_G^T$. We write $T_{\mathfrak{o}}$ for a spanning tree $T$ with a chosen order $\mathfrak{o}$ of its edges. It identfies one such simplex.

Such simplices will each provide one of the lower triangular matrices defining our coactions. If $spt(G)$ is the number of spanning trees of a graph $G$, we get $spt(G)\times e_T!$ such matrices where we use that $e_T=e_G-|G|$ is the same for all spanning trees $T$ of $G$, as there are  $e_T!$ different matrices for each of the $spt(G)$ different $e_T$-cubes $\mathsf{Cub}_G^T$.

\section*{Appendix 2: The lower triangular matrices $M(G,T_{\mathsf{o}})$}
\addcontentsline{toc}{section}{Appendix 2}
Consider a pair $(G,T_{\mathfrak{o}})$ where $G$ is a bridgeless Feynman graph and $T_{\mathfrak{o}}$ a spanning tree $T$ of $G$ with an ordering $\mathfrak{o}$ of its edges $e\in E_T$. There are $e_T!$ such orderings
where $e_T$ is the number of edges of $T$\footnote{In the parametric representation $\mathfrak{o}$ orders them by length.}.

To such a pair we associate a $(e_T+1)\times (e_T+1)$ lower triangular square matrix 
\[
M=M(G,T_{\mathfrak{o}})
\] 
with
$M_{ij}\in H_{GF}$. 

More precisely, $M_{ij}\in \mathsf{Gal}(M)\equiv \mathsf{Gal}(G,T_{\mathfrak{o}})$,
where $\mathsf{Gal}(G,T_{\mathfrak{o}})\subsetneq H_{GF}$ is the set of Galois correspondents of $(G,T_{\mathfrak{o}})$, i.e.\
the graphs  which can be obtained from $G$ by removing or shrinking edges of $T$ in accordance with $\mathfrak{o}$.

As stated above for a pair $(G,T)$ there are $e_T!$ such matrices
$M(G,T_{\mathfrak{o}})$ generated by the corresponding $e_T$-cube of the cubical chain complex associated to any pair $(G;T)\in G_F$ \cite{MarkoDirk}.

$M$ is defined through its entries $M_{ij}\in \mathsf{Gal}(G,T_{\mathfrak{o}})$, $j\leq i$,
\[
M_{ij}:=(G/E_j,(T/E_j\setdiff E^i).
\]
Here $E_j$ is the set given by the first $(e_T-j+1)$-entries of the set 
\[
\{\emptyset, e_{e_T},e_{j-1},\ldots, e_1\}
\]
 and $E^i$ by the first $i$ entries of 
$\{\emptyset, e_1,\ldots,e_{e_T}\}$. We shrink edges in reverse order and remove them in order.

Define the map
$\Delta^M\equiv \Delta^{(G,T_{\mathfrak{o}})}:\mathsf{Gal}(G,T_{\mathfrak{o}})\to \mathsf{Gal}(G,T_{\mathfrak{o}})\otimes \mathsf{Gal}(G,T_{\mathfrak{o}})$,
\be\label{coact}
\Delta^{(G,T_{\mathfrak{o}})}\left(M\right)_{jk}=\sum_{i=1}^{e_T+1}
\left(M\right)_{ik}\otimes \left(M\right)_{ji},
\ee
as before. We often omit the superscript $\{\}^{(G,T_{\mathfrak{o}})}$ when not necessary.
 
Let us now define 
\[
V_1=: \mathsf{Gal}(G,T_{\mathfrak{o}})_/\subsetneq \mathsf{Gal}(G,T_{\mathfrak{o}}),
\]
as the $\mathbb{Q}$-span of elements $M_{j,1}$, $j\geq 2$.

Then we can regard the coproduct $\Delta^M$ as a coaction
\[
\rho_{\Delta^M}:\,\mathsf{Gal}(G,T_{\mathfrak{o}})_/\to
\mathsf{Gal}(G,T_{\mathfrak{o}})_/\otimes \mathsf{Gal}(G,T_{\mathfrak{o}}).
\]

Soon we will evaluate entries in $M_\One$ by Feynman rules.
\[
(M_{\One}){ij}\to \Phi_R(M_{ij})/\Phi_R(M_{jj}). 
\]
This normalization $M\to M_{\mathsf{D}}\times M_\One$ to the leading singularities is common \cite{Brittoetal}.
\section*{ Appendix 3: Summing orders and trees}
\addcontentsline{toc}{section}{Appendix 3}
Let us first consider the sectors we are integrating over.
A graph $G$ provides $e_G!$ sectors. We partition them as follows.
We have $e_G=e_T+|G|$. 
Then
\[
\frac{e_G!}{e_T!\times|G|!}\geq spt(G),
\]
with equality only for $|G|=1$ and the dual of one-loop graphs ('bananas') and $spt(G)$ is the number of spanning trees of $G$ (see also \cite{MarkoDirk}). We note that $e_T!\times |G|!$ is the number of sectors
\[
a_{e_i}\geq a_{e_f} \Leftrightarrow e_i\in E_G\setminus E_T\wedge e_f\in E_T,
\]
where each edge not in the spanning tree is larger than each edge in the spanning tree. This allows to shrink all $e_T$  edges in the spanning tree in any order in accordance with the spine being a deformation retract in the Culler--Vogtmann Outer Space \cite{CullerV}.

The difference 
\[
e_G!-spt(G)\times e_T!\times |G|!
\]
are the sectors where at least one loop shrinks. Any spanning tree $T$ defines a basis of $|G|$ loops $l_i$, $1\leq i\leq |G|$,  provided by a path $p_i$ in $T$ connecting the two ends of an edge
$e_i\in E_G\setminus E_T$. We say that $e_i$ generates $l_i$.

For any given $T$ the sectors where a loop shrinks fulfill two conditions\\
i) for any $l_i$, $a_{e_i}\geq a_e,\,\forall e\in E_{p_i}$,\\
ii) it is not a sector for which 
$a_{e_i}\geq a_{e_f} \Leftrightarrow e_i\in E_G\setminus E_T\wedge e_f\in E_T,
$ holds.\\ The latter condition ii) ensures that when shrinking $e_T$ edges at least one edge in $E_G\setminus E_T$ and hence a loop shrinks. The former condition i) ensures that each loop $l_i$ retracts to its generator $e_i$. 

\begin{ex}
As an example we consider the Dunce's cap and the wheel with three spokes graph.\\
The Dunce's cap: Each spanning tree $T$ gives rise to $2!\times 2!$ sectors $e_T!\times |G|!$. There are five spanning trees, so this covers twenty sectors where no loop shrinks. There are four edges in the Dunce's
cap so we get $4!$ sectors. 
For the four missing sectors four  spanning trees provide one each.\\
The wheel with three spoke graph:\\
$e_T!\times |G|!=3!\times 3!=36$ and there are $16$ spanning trees giving us 576 sectors. The $16$ spanning trees correspond to $16$ choices of three edges while there are $20=({6\atop 3})$ such choices  altogether. There are $6!=720$ sectors.
The missing $144=(20-16)\times 3!\times 3!$ sectors come from the four triangle subgraphs
providing $4\times 3!\times 3!$ sectors.\\
This ends our example.
\end{ex} 
As a result if we let $n(T_{\mathfrak{o}})$ be the number of sectors provided by an ordered spanning tree we have
\begin{lem}
\[
e_G!=\sum_{T\in\mathcal{T}}\sum_{\mathfrak{o}}n(T_{\mathfrak{o}}).
\]
\end{lem}
It thus makes sense to assign a union of sectors $\mathsf{sec}_{T_{\mathfrak{o}}}=\amalg_{j=1}^{n(T_{\mathfrak{o}})}\mathsf{sec}_j$ to each ordered spanning tree $T_{\mathfrak{o}}$. Here $\mathsf{sec}_j\in\mathcal{SEC}_T^{\mathfrak{o}}$, the set of sectors compatible with $T$ and its order of edges $\mathfrak{o}$.
 
We have a coaction $\rho_{\Delta^{T^{\mathfrak{o}}}}$ and coproduct $\Delta^{T^{\mathfrak{o}}}$ for each ordered spanning tree $T^{\mathfrak{o}}$ with a corresponding set $\mathsf{Gal}(G,T_{\mathfrak{o}})$ for each.

We define
\[
\mathsf{Gal}(G):=\amalg_T\amalg_{\mathfrak{o}}\mathsf{Gal}(G,T_{\mathfrak{o}}).
\]
This gives rise to a corresponding matrix $M_G$ formed from $M(G,T_\mathfrak{o})$ and corresponding coproduct and coaction $\Delta^G$.

\end{document}